\documentclass[11pt,a4paper]{article}
\usepackage[margin=1.206in]{geometry} 

\usepackage[latin1]{inputenc}
\usepackage[boxed]{algorithm2e}
\usepackage{amsmath}
\usepackage{amsthm}
\usepackage{mathtools}
\usepackage{amssymb,url}
\usepackage{dsfont}
\usepackage{amscd}
\usepackage{bbm}
\usepackage{graphicx}
\usepackage{multirow}
\usepackage{enumerate}
\usepackage{amscd}

\usepackage{float}
\usepackage{amsthm}
\usepackage{bm}
\usepackage{mathabx}
\usepackage{caption}
\usepackage[list=true]{subcaption}
\usepackage[colorlinks, linkcolor = blue]{hyperref}
\usepackage[compact,raggedright,small]{titlesec}
\usepackage{enumitem}
\usepackage{cleveref}
\usepackage{xcolor}
\usepackage{comment}
\usepackage{mwe}
\usepackage{textcmds}

\usepackage{tikz} 
\usepackage{pgfplots}
\usepackage[subfigure]{tocloft}
\cftpagenumbersoff{figure}
\cftpagenumbersoff{subfigure}
\cftpagenumbersoff{table}
\newcounter{lofdepth}
\newcounter{lotdepth}
\tocloftpagestyle{empty}

\pgfplotsset{compat=newest}
\usetikzlibrary{plotmarks}
\usetikzlibrary{arrows.meta}
\usepgfplotslibrary{patchplots}

\xdefinecolor{tumblue}      {RGB}{  0,101,189}
\xdefinecolor{tumgreen}    {RGB}{162,173,  0}
\xdefinecolor{tumred}        {RGB}{229, 52, 24}
\xdefinecolor{tumivory}      {RGB}{218,215,203}
\xdefinecolor{tumorange}  {RGB}{227,114, 34}
\xdefinecolor{tumlightblue}{RGB}{152,198,234}

\pgfplotsset{plot coordinates/math parser=false}
\newlength\figureheight
\newlength\figurewidth

\DeclareMathOperator{\sign}{sign}

\DeclareMathOperator{\supp}{supp}
\DeclareMathOperator*{\argmin}{arg\,min}

\DeclareMathOperator{\sgn}{sgn}

\DeclareMathOperator{\var}{\mathbb{V}}

\renewcommand{\Im}{\operatorname{Im}}
\renewcommand{\Re}{\operatorname{Re}}
\renewcommand{\Pr}{\mathbb{P}}
\renewcommand{\var}{\operatorname{var}}

\newtheorem{theorem}{Theorem}[section]
\newtheorem{lemma}[theorem]{Lemma}

\newtheorem{remark}[theorem]{Remark}

\newtheorem{definition}[theorem]{Definition}

\usepackage[numbers]{natbib}
\usepackage{cleveref}

\crefformat{equation}{\textup{#2(#1)#3}}
\crefrangeformat{equation}{\textup{#3(#1)#4--#5(#2)#6}}
\crefmultiformat{equation}{\textup{#2(#1)#3}}{ and \textup{#2(#1)#3}}
{, \textup{#2(#1)#3}}{, and \textup{#2(#1)#3}}
\crefrangemultiformat{equation}{\textup{#3(#1)#4--#5(#2)#6}}%
{ and \textup{#3(#1)#4--#5(#2)#6}}{, \textup{#3(#1)#4--#5(#2)#6}}{, and \textup{#3(#1)#4--#5(#2)#6}}

\Crefformat{equation}{#2Equation~\textup{(#1)}#3}
\Crefrangeformat{equation}{Equations~\textup{#3(#1)#4--#5(#2)#6}}
\Crefmultiformat{equation}{Equations~\textup{#2(#1)#3}}{ and \textup{#2(#1)#3}}
{, \textup{#2(#1)#3}}{, and \textup{#2(#1)#3}}
\Crefrangemultiformat{equation}{Equations~\textup{#3(#1)#4--#5(#2)#6}}%
{ and \textup{#3(#1)#4--#5(#2)#6}}{, \textup{#3(#1)#4--#5(#2)#6}}{, and \textup{#3(#1)#4--#5(#2)#6}}

\title{Uncertainty quantification for sparse Fourier recovery}

\author{Frederik Hoppe \thanks{Chair of Mathematics of Information Processing, RWTH Aachen University, Aachen, Germany.} \and Felix Krahmer \thanks{Department of Mathematics and Munich Data Science Institute, Technical University of Munich; Munich Center for Machine Learning, Munich, Germany} \and Claudio Mayrink Verdun \thanks{Department of Mathematics and Department of Electrical and Computer Engineering, Technical University of Munich; Munich Center for Machine Learning, Munich, Germany} \and Marion I. Menzel \thanks{AImotion Bavaria, Faculty of Electrical Engineering and Information Technology, Technische Hochschule Ingolstadt, Ingolstadt, Department of Physics, Technical University of Munich, Garching and GE Healthcare, Munich, Germany.} \and Holger Rauhut \thanks{Department of Mathematics, LMU, Munich,
Germany.}}

\date{\small \today}

\begin{document}
\maketitle
 

\begin{abstract}
One of the most prominent methods for uncertainty quantification in high-dimen-sional statistics is the desparsified LASSO that relies on unconstrained $\ell_1$-minimiza-tion. The majority of initial works focused on real (sub-)Gaussian designs. However, in many applications, such as magnetic resonance imaging (MRI), the measurement process possesses a certain structure due to the nature of the problem. The measurement operator in MRI can be described by a subsampled Fourier matrix. The purpose of this work is to extend the uncertainty quantification process using the desparsified LASSO to design matrices originating from a bounded orthonormal system, which naturally generalizes the subsampled Fourier case and also allows for the treatment of the case where the sparsity basis is not the standard basis. In particular we construct honest confidence intervals for every pixel of an MR image that is sparse in the standard basis provided the number of measurements satisfies $n \gtrsim\max\{ s\log^2 s\log p, s \log^2 p \}$ or that is sparse with respect to the Haar wavelet basis provided a slightly larger number of measurements. 
\end{abstract}

\section{Introduction}

The last decades have seen an explosion of computationally tractable and highly efficient methods for dealing with high-dimensional data. The premise of these methods is that the information contained in many natural data sets relies on statistics of much lower dimension than the original ambient dimension of the data. In many cases, the algorithms to retrieve these data points rely on the fact that they have a parsimonious representation. These ideas form the basis of a change of paradigm in statistics, signal processing and machine learning and became known as \emph{Sparse Regression} (SR) \cite{wainwright2019high,tibshirani1996regression} in the statistical literature or \emph{Compressive Sensing} (CS) \cite{donoho2006compressed, Foucart.2013} in the signal processing literature. The general goal is to estimate a reduced set of dominant components of a high-dimensional object such as a vector, a matrix, or a function, as well as the associated coefficients based on the data.
From the practical point of view, the great success that such techniques have in fields like image processing \cite{adcock2021compressive}, seismology \cite{herrmann2012fighting, oghenekohwo2017low}, medical imaging \cite{lustig2008compressed, ye2019compressed} and radar \cite{ender2010compressive}, just to name a few, illustrates its tremendous impact as a scientific discipline. 

The theoretical underpinning of this field is based on multiple pillars. Firstly, a large number of works concentrate on point estimation (variable selection), aiming to establish sharp oracle inequalities and reconstruction guarantees for the estimators \cite{bunea2007sparsity, SaraA.vandeGeer.2009, bickel2009simultaneous, raskutti2011minimax, bellec2018slope, Buhlmann.2011, verzelen2012minimax}. Secondly, it is important to design and learn efficient representation systems giving rise to sparsity for the high-dimensional objects \cite{tovsic2011dictionary, dumitrescu2018dictionary}. And a third object of intensive study has been the measurement operator used to retrieve the data (also known as the design matrix in the statistical literature), often under constraints arising from applications \cite{baraniuk2008simple}.

Most of this theory, however, is based on a priori understanding of the measurement noise and how it affects high-dimensional estimator. In particular, despite the triumph of sparse recovery techniques in many fields where they are already used at an industrial level \cite{donoho2018blackboard}, a framework that provides uncertainty quantification \cite{smith2013uncertainty} for guiding decision-making in certain critical applications is still missing. Such an understanding would be particularly important for Compressive Sensing applications in the field of medical imaging, such as Magnetic Resonance Imaging (MRI), one of the main motivations behind CS theory \cite{candes2006near, donoho2006compressed, lustig2008compressed}. Since reliable medical imaging procedures are pivotal for accurate interpretation and diagnostic tasks, it is of fundamental importance to develop a theory that quantifies the quality of such images based on Sparse Regression. 

Even for (sub-)Gaussian designs, which are typically easier to analyze, the important question \emph{given a solution estimator of an underdetermined sparse regression problem, how can one quantify its confidence levels?} remained open for a long period. In the classical setting, where the number of available samples $n$ is much higher than the ambient dimension $p$, in contrast, the exact or asymptotic distribution of the retrieved solution is known for many noise distributions \cite{lehmann2006theory}. Therefore, procedures for uncertainty quantification (UQ) are largely available \cite{lehmann2005testing}.

The reason why these methods do not carry over to the underdetermined case in high dimensions is that good estimators for such problems are basically always non-linear and it is hard to obtain a tractable characterization of the probability distribution of the parameters of interest for such non-linear estimators \cite{fu2000asymptotics, potscher2009distribution}.

In particular, this is the case for estimators based on the $\ell_1$-norm, such as the LASSO \cite{tibshirani1996regression, tibshirani2011regression} and the quadratically-constrained basis pursuit \cite{van2009probing}, which are arguably the most well-studied and simplest class of approaches for variable selection in underdetermined high-dimensional problems. More precisely, the sparsity enforcement of the selection and shrinkage operations has been shown to introduce a bias that distorts the probability distribution in a non-linear way \cite{zhang2008sparsity}. This bias, together with the intractability of the distribution of its solution, makes the construction of confidence intervals and the performance of hypothesis testing very challenging.

Recently, a series of papers initiated a desparsified approach to sparse regression \cite{Zhang.2014, Javanmard.2014, Javanmard.2018, vandeGeer.2014}. This technique is able to characterize the distribution of a modified estimator based on the KKT conditions of the LASSO solution. For this modified estimator, sharp confidence intervals were derived for variable selection. For i.i.d. (sub-)Gaussian designs with known covariance, by assuming that the object to be retrieved is $s$-sparse, the weakest known constraint on the admissible sparsity level is  $s = o(n / \log^2 p)$ \cite{Javanmard.2018}, getting very close to the provably optimal threshold for sparse estimation without uncertainty quantification, namely $s = o(n / \log p)$ \cite{wainwright2019high}. 

Although these results provided fundamental theoretical insight for the theory of uncertainty quantification in the high-dimensional regime, such fully random matrices are of limited practical use in signal processing. In MRI, for example, the measurement process is highly structured and can be described by a (subsampled) Fourier operator imposed by the physics behind the acquisition procedure \cite{zhi2000principles}. Moreover, structured matrices often allow for faster algorithmic processing by exploiting the fast Fourier transform (FFT) for matrix-multiplication as well as efficient storage. 
The only paper that considers a more general class of designs admitting such structured examples is the groundbreaking analysis provided in \cite{vandeGeer.2014}. The wider range of applicability of their approach, however, comes at the expense of a suboptimal constraint on the sparsity level of $ s = o (\sqrt{n}/\log p)$, hence somewhat limiting the applicability of this analysis for the subsampled Fourier scenario.

\subsection{Our contribution}

In this work, we address the described limitations of previous works establishing honest confidence intervals, in the sense of \cite{li1989honest}, for the desparsified LASSO estimator with matrices generated from bounded orthonormal systems (BOS) in the near-optimal regime $n \gtrsim \max\{ s\log^2 s\log p, s \log^2 p \}$. Our results can be applied to MRI problems, where a particular instance of BOS, namely, a subsampled Fourier matrix, is used. Our main result takes the following form

\begin{theorem}\label{thm:informal2} (Informal version) Let $\frac{1}{\sqrt{n}} X$ be a normalized random sampling matrix associated to a Bounded Orthonormal System with constant $K\geq 1$, defined in Definition \ref{def:matrix_BOS}, and a sample covariance matrix $\hat{\Sigma}=X^*X/n$. Suppose that the data $y \in \mathbb{C}^n$ is given by the model $y=X\beta^0+\varepsilon$, with $\varepsilon \sim\mathcal{CN}(0,\sigma^2 I_{n\times n})$, where $\beta^0$ is assumed to be an $s$-sparse vector. For n $ \gtrsim \max\{ s\log^2 s\log p, s \log^2 p \} $ and a given consistent noise estimator $\hat{\sigma}$, the confidence region with significance level $\alpha\in(0,1)$, for $\beta_i^0\in\mathbb{C}$ estimated via the desparsified LASSO $\hat{\beta}^u$ in \eqref{eq:desparsified_LASSO}, i.e.,
\begin{equation*}
    J^{\circ}_i(\alpha):=\{z\in\mathbb{C}:\vert\hat{\beta}_i^u-z\vert\leq\delta_i^{\circ}(\alpha)\},\qquad \text{with} \quad \delta^{\circ}(\alpha):=\frac{\hat{\sigma}\hat{\Sigma}_{ii}^{1/2}}{\sqrt{n}}\sqrt{\log\left(\frac{1}{\alpha}\right)}
\end{equation*}
is asymptotically valid, i.e.,
$$\lim\limits_{n\to\infty}\mathbb{P}\left(\beta_i^0\in J^{\circ}_i(\alpha)\right)=1-\alpha.$$
\end{theorem}

Our approach is general enough to include important sampling schemes such as Fourier measurements, as they appear in MRI applications, or Hadamard matrices, as well as local incoherent bases that are used in many applications where the object to be retrieved is sparse in a basis other than the canonical basis. In particular, for s growing sufficiently slow as compared to $p$, i.e., $\log p > \log^2 s$, our results require $s = o(n / \log^2 p)$, which matches the theoretically optimal result from the literature, developed for Gaussian designs \cite{Javanmard.2018}. Moreover, the size of the confidence region is sharp, in the minimax sense, as a function of $n$ \cite{Cai.2017}. This will be discussed in Section \ref{sec:confidence regions}.

To our knowledge, this is the first time that such a near-optimal result on confidence bounds has been derived for any class of structured design matrices. We state and prove our results for the complex LASSO since, in certain scenarios, such as in MRI, the measurements are of a complex nature. 

Moreover, to the best of our knowledge, this is the first work to extend the desparsified LASSO estimator to the case where the ground truth is not sparse but can instead be sparsified by an orthogonal transform. This setting is crucial in many applications, such as medical imaging, where the underlying data is typically sparse when represented on a Wavelet basis or by using Total Variation (TV) regularizers. In particular, we illustrate this idea in Section \ref{sec:desparsified_haar} with the Haar wavelet transform.

\subsection{Notation}
We define a complex number $z\in\mathbb{C}$ as $z=\Re(z)+i\Im(z)$ with $\Re(z),\Im(z)\in\mathbb{R}$. For a natural number $p\in\mathbb{N}$ we write $[p]:=\{1,\hdots,p\}$. For two numbers $a,b \in \mathbb{R}$, we denote by $a \gtrsim b$ the fact that there is a constant $C>0$, independent from the dimensions, such that $a \geq C b$. We refer to the $i$th row of a matrix $X\in\mathbb{C}^{n\times p}$ as $x_i^T$, $x_i\in\mathbb{C}^p$, and to the $j$th column as $X_j\in\mathbb{C}^n$ while the $ij$-entry is denoted by $X_{ij}\in\mathbb{C}$. The matrix $X_{-j}$ is the matrix $X$ where the $j$th column is replaced by a column consisting of zeros. The Hermitian adjoint of a matrix $X$ is denoted by $X^*$. The $i$-th canonical basis vector is denoted by $e_i=(0,\hdots,0,1,0,\hdots,0)^T$. The complex inner product of two vectors $v,w\in\mathbb{C}^p$ is defined as $\langle v,w\rangle=v^*w=\sum_{i=1}^p\overline{v}_iw_i$. The $\ell_2$ norm of $v$ is $\Vert v\Vert_2=\sqrt{\langle v,v\rangle}$. Furthermore, the $\ell_1$ norm is defined to be $\Vert v\Vert_1=\sum_{i=1}^p\vert v_i\vert$ and the $\ell_{\infty}$ norm is $\Vert v\Vert_{\infty}=\max_{i\in[p]}\vert v_i\vert$. A vector $v$ is called $s$-sparse if it has at most $s$ nonzero entries. We define the sample covariance matrix of a matrix $X$ by $\hat{\Sigma}=X^*X/n$. For a given vector $x \in \mathbb{C}^n$, we denote by $sign(x)$ the vector that returns the unimodular part of each entry of $x$, i.e., for a given entry of the form $x_i=r e^{i\theta}$, it returns $sign(x)_i=e^{i\theta}$. We denote the expected value and the variance of a random variable $X$ by $\mathbb{E}$ and $\var$, respectively.

\section{Background}\label{sec:desparsified LASSO}

In this section, we describe the uncertainty quantification problem for high-dimensional sparse regression, the desparsified technique to address this problem and the relevant background for bounded orthonormal systems.

\subsection{Sparse regression}

For a design/measurement matrix $X \in \mathbb{C}^{n \times p}$ with rows $ x_1^T, \dots, x_n^T$ and a data vector $ y = (y_1, \dots y_n) \in \mathbb{C}^{n} $, we are interested in the regression model

\begin{equation}\label{Cmodel}
y=X\beta^0+\varepsilon, 
\end{equation}
where the noise vector $\varepsilon\sim\mathcal{CN}(0,\sigma^2 I_{n\times n})$ is assumed to be a complex standard Gaussian vector whose components $\varepsilon_i$ are independent. Note that we are considering complex-valued representations of MRI measurements as our main motivating example. 

We consider the high-dimensional setting where $ n \ll p$ and we assume that the ground-truth $\beta^0$ is $s_0$-sparse, i.e., the cardinality of the support $S_0 := \{ i \ | \ \beta_i \neq 0\}$ of $\beta^0$ is $s_0\ll p$. The main goal is to estimate $\beta^0 \in \mathbb{C}^p$ as well as to provide confidence regions for $\beta^0$ based on this estimator. 
A point estimator is a map that provides an estimate for the ground truth $\beta_0$ for given data $y \in \mathbb{C}^n$. A common regularizer to enforce the sparsity model is the $\ell_1$-norm $\|z\|_1 = \sum |z_j|$. A natural estimator corresponding to this regularizer is the LASSO \cite{tibshirani1996regression, maleki2013asymptotic}, which is mathematically described by
\begin{equation}\label{eq:LASSO}
 \hat{\beta} = \argmin\limits_{\beta\in\mathbb{C}^p}\frac{1}{2n}\Vert X\beta-y\Vert_2^2+\lambda\Vert \beta\Vert_1,
\end{equation}
where $\lambda \in \mathbb{R}$ is a tuning parameter to balance data fidelity and sparsity. Here, we assume that \eqref{eq:LASSO} has a unique solution. By \cite[Theorem 1]{Zhang.2015}, the so-called Fuchs condition is necessary and sufficient for uniqueness, i.e., when $A_S$ is of full column rank and there exists a $y \in \mathbb{C}^n$ such that $A^*_S y=s$ and $\|A_{S^c}^* \|_{\infty} <1$, where $S = \supp (\beta^0)$ and $s = \sign(\beta^0_S)$. See also \cite{tropp2005recovery, fuchs2004sparse}.

To stress the fact that all parameters are complex-valued, this problem is often referred to as the complex LASSO (c-LASSO). This complex version of the LASSO is also relevant in other applications such as audio processing \cite{lilis2010sound} and direction of arrival estimation problems \cite{mecklenbrauker2017c, ollila2016direction}. Accurate
formulas for the phase transition and noise sensitivity that determine the fundamental recovery limit of the c-LASSO have been established in \cite{maleki2013asymptotic}. We note that the real LASSO case can be seen as a particular instance of the complex one and the results can be translated one-to-one.

The c-LASSO can be conveniently expressed as a real-valued problem by splitting the complex measurement matrix $X \in \mathbb{C}^{n \times p}$ and the vectors $\beta \in \mathbb{C}^p$, $y \in \mathbb{C}^n$, into their real and imaginary parts, i.e., 

$$\tilde{X}:=\begin{pmatrix}\Re X & -\Im X\\ \Im X & \Re X\end{pmatrix},\,\tilde{\beta}:=\binom{\Re\beta}{\Im\beta},\text{ and }\tilde{y}:=\binom{\Re y}{\Im y}.$$

This yields the equivalent formulation

\begin{equation}\label{eq:classo}
    \min\limits_{\tilde{\beta}\in\mathbb{R}^{2p}}\frac{1}{2n}\Vert \tilde{y}-\tilde{X}\tilde{\beta}\Vert_2^2+\lambda\sum\limits_{i=1}^p\sqrt{\tilde{\beta}_i^2+\tilde{\beta}_{i+p}^2}.
\end{equation}

This formulation is a particular instance of the group LASSO \cite{Yuan.2006} where one has $p$ groups of size 2. We note that the theory developed in this paper can be extended to the group LASSO or sparse group LASSO problem for groups of any size by following the proof of Theorem \ref{thm:mainresult} with an appropriate $\ell_2$-bound for the underlying optimization problem. For a discussion on such bounds, see, e.g., \cite{li2023sharp}.

\subsection{Quantifying statistical uncertainty}
In the field of uncertainty quantification, tools from probability theory and statistics are used to quantify uncertainties in mathematical models or in real-world experiments. In this paper, we aim at quantifying the unknown ground truth of model \eqref{Cmodel} using methods from statistical inference based on data whose sampling is affected by noise. One of the main methods in statistical inference is the construction of a point estimator, e.g., the LASSO. We will go beyond the basic statistical method of point estimators and consider the construction of confidence intervals, which provide \emph{``error bars''} around a point estimator.

\begin{definition}\cite[Definition 6.12]{wasserman2013all}
A $(1-\alpha)$ confidence interval for a parameter $\beta^0\in\mathbb{R}$ is an interval $C_n(a,b)$, whose size depends on $n$ and its extreme points $a$ and $b$ depend on the data $y_1,\hdots,y_n$, such that
$$\mathbb{P}(\beta^0\in C_n(a,b))\geq 1-\alpha\qquad \forall\beta^0\in B,$$
where $B$ is the parameter space of the underlying statistical model.
\end{definition}

In many cases, point estimators are asymptotically normal, i.e., $\hat{\beta}\overset{d}{\to}\mathcal{N}(\beta^0,\hat{\sigma}^2)$, where $\overset{d}{\to}$ means convergence in distribution as the problem dimension $n$ tends to infinity. In this setting, there is a standard procedure to derive confidence intervals from point estimators.

\begin{theorem}\cite[Theorem 6.16]{wasserman2013all}\label{thm:conf_int_with_gaussian}
Suppose that $\hat{\beta}\overset{d}{\to}\mathcal{N}(\beta^0,\hat{\sigma}^2)$ as $n\to\infty$, where $\overset{d}{\to}$ means convergence in distribution. Let $\Phi^{-1}$ be the quantile function of a standard Gaussian. Set $C_n = (\hat{\beta}-\Phi^{-1}(1-\alpha/2)\hat{\sigma},\hat{\beta}+\Phi^{-1}(1-\alpha/2)\hat{\sigma}).$ Then $\lim\limits_{n\to\infty}\mathbb{P}(\beta^0\in C_n)= 1-\alpha$.
\end{theorem}

In the classical setting, where $n > p$, the central limit theorem assures that the sample mean is an unbiased estimator that is Gaussian distributed \cite[Theorem B.97]{schervish2012theory}. However, in the high-dimensional setting $p\gg n$ and in particular in sparse regression, the distribution of point estimators is usually not available. Therefore, alternative approaches are required. The first techniques to perform uncertainty quantification of sparse regression problems rely on sample-splitting methods and appeared in \cite{wasserman2009high, meinshausen2009p, shah2013variable}. The idea is to use half of the samples for solving the sparse regression problem and the other half to perform statistical inference via ordinary least-squares and t-tests. However, they require strong assumptions on the regression parameters to be retrieved, e.g., the minimum value attained in the support, the so-called \emph{$\beta$-min} assumption \cite{dezeure2015high}. This quantity is not known a priori, and significance tests are usually designed to infer it, i.e., to establish if a regression coefficient is above a certain threshold \cite{dezeure2015high}. Another method that provides uncertainty quantification is the ridge projection method \cite{buhlmann2013statistical}. Unlike the previous approach, it does not require any assumption on the measurement matrix or a lower bound for the coefficients in the support of the vector to be retrieved. However, such a method does not reach the asymptotic Cram\'er-Rao bound and is not able to provide optimal confidence intervals \cite{dezeure2015high}. A very promising estimator that satisfies the assumption of Theorem \ref{thm:conf_int_with_gaussian}, is the desparsified LASSO as described next.

\subsection{The desparsified LASSO}\label{subsec:desparsified_LASSO}

Following the works \cite{Javanmard.2014, Javanmard.2018, vandeGeer.2014}, we aim to derive confidence bounds for the c-LASSO estimator in the case that the design matrix is given by a random sampling matrix from \emph{bounded orthonormal systems} - generalizing the Fourier system - discussed in more detail in Section \ref{sec:BOS}. 

Most previous contributions assume the design to be a real (sub-)Gaussian matrix $X\in\mathbb{R}^{n\times p}$, i.e., a matrix with light-tailed distribution \cite{Javanmard.2014, Javanmard.2018, vandeGeer.2014}; only the work \cite{vandeGeer.2014} also provides results for fixed (deterministic) designs $X$ and bounded random designs under strong assumptions. The construction of a desparsified estimator for the LASSO solution $\hat{\beta}$, as described by \cite{vandeGeer.2014}, is based on the KKT conditions and reads as
\begin{equation}
\begin{split}
    & - X^*(y-X\hat{\beta})/n + \lambda \hat{\kappa} = 0\\
    & ||\hat{\kappa}||_{\infty} \leq 1 \quad \text{and} \quad \hat{\kappa}_j = \sgn (\hat{\beta}_j) \ \text{if} \ \hat{\beta}_j \neq 0,
\end{split}
\end{equation}
where $\hat{\kappa}$ is an element of the subdifferential of the $\ell_1$-norm. Writing $\hat{\Sigma}=X^*X/n$, the equation above can be rewritten as 
\begin{equation}
\begin{split}
     \lambda \hat{\kappa} + \hat{\Sigma}(\hat{\beta} - \beta^0) = X^* \varepsilon /n.
\end{split}
\end{equation}
Denoting by $M$ an approximate inverse of the empirical covariance matrix $\hat{\Sigma}$, i.e., a matrix (to be chosen) such that $M \hat{\Sigma} \approx I_{p\times p}$, we have
\begin{equation}
\begin{split}
     \hat{\beta} - \beta^0 + M \lambda \hat{\kappa} = M X^* \varepsilon /n - (M \hat{\Sigma} - I_{p\times p})(\hat{\beta} - \beta^0).
\end{split}
\end{equation}

One of the important achievements of the desparsified LASSO theory shows that the remainder term $(M \hat{\Sigma} - I_{p\times p})(\hat{\beta} - \beta^0) $ asymptotically vanishes since the term $  M \lambda \hat{\kappa}$ compensates for the bias introduced by the $\ell_1$ regularizer. Since $ M X^* \varepsilon /\sqrt{n} \sim\mathcal{N}(0,\sigma^2\hat{\Sigma}) $, this allows us to construct pointwise confidence intervals for $\beta_j^0$. Therefore, it is reasonable to add a subgradient of the $\ell_1$-norm at the LASSO solution $\hat{\beta}$ and to consider the \emph{desparsified LASSO} estimator defined as

\begin{equation}\label{eq:desparsified_LASSO}
    \hat{\beta}^u= \hat{\beta} + M \lambda \hat{\kappa} = \hat{\beta}+\frac{1}{n}MX^*(y-X\hat{\beta}).
\end{equation}

The previous approaches estimated the terms $ (M \hat{\Sigma} - I_{p\times p})$ and $ (\hat{\beta} - \beta^0) $ separately via an $\ell_1-\ell_{\infty}$ argument that leads to non-optimal bounds. The only exception is the seminal paper \cite{Javanmard.2018} that, instead, uses a leave-one-out argument and strongly exploits the independence of $X\Sigma^{-1}e_i$ and $X_{-i}$, which holds, for example, for matrices with Gaussian rows, but unfortunately does not hold for heavy-tailed matrices such as those generated from bounded orthonormal systems. Our main result, Theorem \eqref{thm:mainresult}, extends the applicability to heavy-tailed matrices and provides (near-) optimal bounds.

\subsection{The desparsified approach to sparse regression.}
Without any prior assumptions about the signal strength beyond sparsity, a low dimensional projection (LDP) approach for confidence estimation is developed for the LASSO estimator in \cite{Zhang.2014}. This is, to the best of the authors' knowledge, the first work towards a theory for assigning uncertainty in high-dimensional sparse models that has minimal assumptions on the quantities to be estimated. This approach was extended in \cite{Javanmard.2014, Javanmard.2018, vandeGeer.2014} as described in the previous sections, and is now known under the name \emph{de-sparsified LASSO}. The idea of these contributions, as explained in Section \ref{subsec:desparsified_LASSO}, is to create a non-sparse estimator based on the KKT conditions of the LASSO formulation. This new estimator reduces the shrinkage bias introduced by the LASSO and, for this reason, it is also known as \emph{desparsified LASSO} in the literature. The works \cite{vandeGeer.2014,Javanmard.2014} constructed a desparsified LASSO estimator provided that the sample size is at least $n \gtrsim s^2 \log^2 p$ and that the covariance matrix is sparse, which is suboptimal in $s$. For Gaussian designs with known covariance $\Sigma=I_{p\times p}$, \cite{javanmard2014hypothesis} established that it is sufficient to have $n \gtrsim  s \log(p/s)$ measurements in order to create confidence intervals of length $n^{-1/2}$ from the desparsified LASSO. This result, however, holds in a weak sense and implies coverage of the constructed confidence intervals only on average over the coordinates $i \in \{1 , \dots, p\} $. Later, \cite{Javanmard.2018} improved on this result by requiring the nearly optimal condition $n \gtrsim s \log^2 p$ in the case $X$ is Gaussian with known variance or, at least, by assuming that the inverse of the covariance matrix is sparse as in the assumptions made in \cite{vandeGeer.2014}.
The same paper also introduces an alternative proof strategy based on sample splitting, which, however, does not require additional assumptions as they appear in the early approaches discussed above. We will follow a similar proof strategy in this paper.

Further, \cite{Cai.2017} studied the expected length of confidence intervals for the desparsified LASSO in the oracle setting where the sparsity level is known. It is established that the minimax expected length is of the order $1/\sqrt{n}+ s \frac{\log p}{n}$ of confidence intervals for $\beta_i$. The paper constructs confidence intervals in the moderately sparse regime $\frac{\sqrt{n}}{\log p} \ll s \lesssim \frac{n}{\log p}$. Recently, \cite{shinkyu2022small} established that $n \gtrsim s^2 \log p$ is sufficient for the asymptotic normality of the desparsified LASSO, in the case where $p \geq C n$, without imposing any sparsity assumption on the precision matrix, i.e., the inverse of the population covariance matrix. Furthermore, \cite{bellec2019biasing} established asymptotic normality results for a de-biased estimator of convex regularizers beyond the $\ell_1$-norm, in the regime $p/n \rightarrow \gamma \in (0, +\infty)$ and $s/n \rightarrow \kappa \in (0,1)$.

\subsection{Bounded orthonormal systems}\label{sec:BOS}

As described in the introduction, so far, most works on desparsified estimators are restricted to random (sub-)Gaussian measurement matrices that do not possess any structure. For bounded designs, the only available result requires a sample complexity scaling like $n \gtrsim  s^2 \log^2 p$ \cite{vandeGeer.2014}, which is a suboptimal sample complexity. However, one of the central motivating applications of compressed sensing theory is magnetic resonance imaging (MRI) and there, as a result of the Bloch equation that models the magnetic resonance phenomenon \cite{zhi2000principles}, one needs to consider subsampled Fourier measurements instead \cite{rudelson2008sparse}.

\begin{definition}\label{def:partial Fourier}
A Fourier matrix $F\in\mathbb{C}^{p\times p}$ is defined entrywise as $F_{l,k}=e^{2\pi i(l-1)(k-1)/p}$, $l,k\in[p]$. The random subsampled Fourier matrix $F_{\Omega}$ consists of the rows $f_{t_1},..., f_{t_n}$ from the original Fourier matrix $F$, where the indices $\Omega:=\{t_1,\hdots,t_n\}$ are sampled independently and uniformly from $[p]$.
\end{definition}

Our goal in this paper is to close this gap and provide guarantees with near-optimal sample complexity for bounded orthonormal systems, a class of structured random matrices that includes random subsampled Fourier matrices as a particular case and also covers structured random matrices emerging when band-limited functions are to be reconstructed from few random samples. 

A bounded orthonormal system consists of a family of functions that are orthonormal in the $L_2$ norm and uniformly bounded in the $L_{\infty}$ norm. In the following exposition, we follow \cite[Section 12.1]{Foucart.2013}. Further discussions can be found in \cite{rauhut2016interpolation, andersson2014theorem, brugiapaglia2018robustness, xu2020unraveling}.

\begin{definition}
Let $\nu$ be a probability measure on $\mathcal{D}\subset\mathbb{R}^d$. A bounded orthonormal system (BOS) with constant $K\geq 1$ consists of complex-valued functions $\{\phi_1,\hdots,\phi_p\}$ on $\mathcal{D}$ that have the properties
\begin{equation}\label{eq:orthonormal_system}
    \int_{\mathcal{D}}\phi_j(t)\overline{\phi_k(t)}d\nu(t)=\delta_{j,k}\qquad\forall j,k\in[p]
\end{equation}
and
\begin{equation}\label{eq:bounded_orthonormal_system}
    \Vert \phi_j\Vert_{\infty}:=\sup\limits_{t\in\mathcal{D}}\vert\phi_j(t)\vert\leq K\qquad\forall j\in[p].
\end{equation}
\end{definition}
The crucial point of the boundedness condition \eqref{eq:bounded_orthonormal_system} is that $K$ should ideally be independent of $p$.
Based on this family of functions, one defines a random matrix by evaluating the BOS at $n$ randomly chosen sample points. 

\begin{definition}\label{def:matrix_BOS}
A matrix $X\in  \mathbb{C}^{n\times p}$ is said to be a random sampling matrix associated to a BOS if its entries are
$$X_{lk}=\phi_k(t_l),\qquad l\in[n],\,k\in[p],$$
where $t_1,\hdots ,t_n\in\mathcal{D}$ are chosen independently at random according to $\nu$.
\end{definition}

Note that due to the independence of the sampling points $t_1,\hdots ,t_n$, the rows of the matrix $X$ are independent. In the case of subsampled Fourier matrices $F_{\Omega}$, the functions consist of the trigonometric system $\phi_k(t)= e^{2 \pi i kt} $ and the uniform bound in condition \eqref{eq:bounded_orthonormal_system} is $K=1$. A random sampling matrix associated to a BOS has stochastically independent rows, although the entries within each row are not independent.

One of the main conditions on the measurement matrix required to establish consistency of the LASSO estimator and that permeates the whole field of sparse regression \cite{SaraA.vandeGeer.2009} is the following.

\begin{definition}\label{defi:RIP}
A matrix $X$ satisfies the restricted isometry property (RIP) of order $1\leq s\leq p$ if there is a constant $\delta_s\in(0,1)$ such that
$$(1-\delta_s)\Vert \beta\Vert_2^2\leq \Vert X\beta\Vert_2^2\leq (1+\delta_s)\Vert \beta\Vert_2^2$$
for all $s$-sparse vectors $\beta\in \mathbb{C}^p$.
\end{definition}

In this paper, we will make use of the following bound on the RIP constants for random sampling matrices associated to a BOS \cite{cheraghchi2013restricted, rauhut2016interpolation}. 
Although not being described by a light-tailed probabilistic model, random matrices sampled from a BOS still act as quasi-isometries on the subset of sparse vectors, i.e., random matrices sampled from BOS satisfy the restricted isometry property \cite{BRUGIAPAGLIA2021231}. 

\begin{theorem}\cite[Theorem 1.1 and Theorem 2.3]{BRUGIAPAGLIA2021231}\label{thm:rip for BOS}
Let $X \in \mathbb{C}^{n \times p}$ be the random sampling matrix associated to a BOS with constant $K \geq 1$. If, for $\delta \in (0,1)$,
\begin{equation}\label{eq:n for RIP}
    n \geq C K^2 \delta^{-2} s \log^2(sK^2/\delta) \log(ep)
\end{equation}
for a certain positive constant $C>0$, then with probability at least $1- 2\exp( -\delta^2 n/(sK^2))$, the matrix $\frac{1}{\sqrt{n}} X$ is an RIP matrix of order $s$ with constant $\delta_s \leq \delta $.
\end{theorem}

As a consequence of the sampling pattern, the second moment matrix of any row $x_j$ is the identity $I_{p\times p}$:
\begin{align*}
    \mathbb{E}[x_jx_j^*]_{kl}=\mathbb{E}[\phi_k(t_j)\overline{\phi_l(t_j)}]=\int_{\mathcal{D}}\phi_k(t)\overline{\phi_k(t)}d\nu(t)=\delta_{kl}.
\end{align*}

\begin{lemma}\label{le:expectation_sample_cov}
Let $X$ be a random sampling matrix associated to a BOS with constant $K\geq 1$. Then the sample covariance matrix $\hat{\Sigma}=\frac{X^*X}{n}$ satisfies
$$\mathbb{E}\big[\hat{\Sigma}\big]=I_{p\times p}.$$
\begin{proof}
By definition
\begin{align*}
    \mathbb{E}\big[\hat{\Sigma}_{ij}\big]=\frac{1}{n}\sum\limits_{l=1}^n\mathbb{E}[X_{li}\overline{X}_{lj}]=\frac{1}{n}\sum\limits_{l=1}^n\mathbb{E}[\phi_i(t_l)\overline{\phi_j(t_l)}]=\frac{1}{n}\sum\limits_{l=1}^n\int_{\mathcal{D}}\phi_i(t)\overline{\phi_j(t)}d\nu(t).
\end{align*}
Since $(\phi_k)_{k\in[p]}$ forms an orthonormal system in $(\mathcal{D},\nu)$, we have $\mathbb{E}\big[\hat{\Sigma}_{ij}\big]=\frac{1}{n}\sum\limits_{l=1}^n\delta_{ij}=\delta_{ij}.$
\end{proof}
\end{lemma}
Hence, we expect the sample covariance to be close to the identity. The estimate
$$\vert\hat{\Sigma}_{ij}\vert=\frac{\vert\langle X_i,X_j\rangle\vert}{n}=\frac{1}{n}\left\vert\sum\limits_{k=1}^n X_{ik}\overline{X}_{jk}\right\vert\leq K^2,\qquad i,j\in[p]$$
shows that the entries of the sample covariance are restricted to the range $[0,K^2]$. 
This will be crucial for applying Bernstein's inequality when estimating the deviation of $M\hat{\Sigma}-I_{p\times p}$.
To make this term vanish in expectation, we will choose
$M=I_{p\times p}$.

In many practical applications, however, the underlying signal to be retrieved is often not sparse in the canonical basis. In particular, in MRI applications, the underlying image is not necessarily sparse, but it can rather be sparsified with respect to another basis. Therefore, one key assumption is the existence of a sparsifying transform, e.g., an orthogonal operator $T:\mathbb{C}^p\to\mathbb{C}^p$ mapping the non-sparse image $\beta^0\in\mathbb{C}^p$ to an $s_0$-sparse signal $z^0=T\beta^0$. In this case, the underlying model becomes $y=XT^{-1}z^0+\varepsilon$, where $X\in\mathbb{C}^{n\times p}$ is a random sampling matrix associated to a BOS with constant $K\geq 1$, $y\in\mathbb{C}^n$ the measured data and $\varepsilon\sim\mathcal{CN}(0,\sigma^2I_{n\times n})$. In order to apply the existing theory, we need to guarantee that the matrix $XT^{-1}$, whose rows are orthogonal, satisfies a version of \eqref{eq:bounded_orthonormal_system}, namely
\begin{equation}\label{eq:boundedness_sparse_transform}
    \max\limits_{i,j\in[p]}\vert (XT^*)_{ij}\vert = \max\limits_{i,j\in[p]}\vert \langle  x_i,t_j\rangle\vert\leq L \text{ for a constant } L\geq 1.
\end{equation}
This is the same as saying that $XT^{-1}$ is a random sampling matrix associated to a BOS with constant $L\geq 1$. Theorem 1 in \cite{6651836} shows that if the rows of $XT^{-1}$ are not selected uniformly but rather with a modified probability measure, then the matrix $XT^{-1}$ is a matrix associated with a BOS and, furthermore, that RIP holds with high probability provided that the number of samples is high enough. The result, whose sample complexity assumption is stated in an improved way here for the first time due to the use of Theorem \ref{thm:rip for BOS} instead of \cite{Rauhut+2010+1+92} as used in the original proof, reads as follows.

\begin{theorem}\cite[Theorem 1]{6651836}\label{thm:sparse_transform_RIP}
Let $\{x_1,\hdots,x_p\}$ and $\{t_1,\hdots,t_p\}$ be orthonormal bases of $\mathbb{C}^p$. Assume, that $\sup\limits_{1\leq k\leq p}\vert \langle x_j,t_k\rangle \vert\leq \kappa_j.$ Let $s\geq \log p$, suppose
$$n\geq C\delta^{-2}\Vert \kappa\Vert_2^2 s \log^2 s \log p,$$
and choose $n$ (possibly not distinct) indices $j\in\Omega\subset [p]$ i.i.d. from the probability measure $\nu$ on $[p]$ given by
$$\nu(j)=\frac{\kappa_j^2}{\Vert \kappa\Vert_2^2}.$$
Consider the matrix $A\in\mathbb{C}^{n\times p}$ with entries
$$A_{jk}=\langle x_j,t_k\rangle, \qquad j\in\Omega, k\in[p],$$
and consider the diagonal matrix $D=\operatorname{diag}(d)\in\mathbb{C}^{p\times p}$ with $d_j=\Vert \kappa\Vert_2/\kappa_j$. Then with probability at least $1-p^{-c\log^3 s}$, the restricted isometry constant $\delta_s$ of the preconditioned matrix $\frac{1}{\sqrt{n}}DA$ satisfies $\delta_s\leq\delta$.
\end{theorem}

In the case that the sparsifying transform $T$ is given by a Haar wavelet basis \cite{6651836} and the measurement matrix is given by a subsampled Fourier transform, a consequence of the theorem above is that the matrix $\frac{1}{\sqrt{n}}DF_{\Omega} T^*$ is a matrix associated with a BOS with constant $\Vert \kappa\Vert_2\leq c \sqrt{\log_2 p}$ with a constant $c_1>0$ \cite[Theorem 4]{6651836}. Thus, in this case, the price for applying the sparsifying transform when the image is not sparse in the canonical basis is an additional $\log_2 p$-factor in the necessary number of measurements.
Furthermore, the main result of \cite{6651836} shows that if $n\geq C s\log^3s\log^2p$ then the recovery of $\beta^0$ via $\ell_1$-minimization from the measurements
$y=F_{\Omega}\beta^0+\varepsilon,$ where $\Omega\subset[p]$ is sampled with respect to $\nu$, approximates $\beta^0$ by $\hat{\beta}$
with high probability:
$$\Vert \beta^0-\hat{\beta}\Vert_2\leq C \frac{\Vert H\beta^0- (H\beta^0)_{s_0}\Vert_2}{\sqrt{s_0}}+\frac{\Vert D\varepsilon\Vert_2}{\sqrt{n}}$$

We will use the results above in order to develop an uncertainty quantification theory and construct confidence intervals in the case that images can be sparsely represented by a Haar wavelet transform. One difference from our work, however, is that the work \cite{6651836} considered a normalized trigonometric system $\phi_k(t)= \frac{1}{\sqrt{p}} e^{2 \pi i kt}$. Here in this paper, instead, we will consider the Fourier basis without normalization, i.e., $\phi_k(t)= e^{2 \pi i kt}$. This means that the weights to be considered for the matrix $\frac{1}{\sqrt{n}}DA$ will be given by $d_j=\Vert \kappa\Vert_2/(\sqrt{p} \kappa_j)$. See Section \ref{sec:desparsified_haar} for further discussion.

\section{Related Work}
The theory for the desparsified LASSO has been extended in many directions and it is a rapidly growing field. Due to the large number of contributions, we can only offer a brief, non-exhaustive overview. The works \cite{ren2015asymptotic} and \cite{jankova2018inference} proposed generalizations of the desparsified LASSO for Gaussian graphical models, and \cite{fang2017testing} considered the desparsified method in high-dimensional Cox models. Moreover, \cite{buhlmann2015high} studied the question when the desparsified LASSO procedure is valid for the construction of statistical hypothesis tests and confidence intervals in the case of misspecified high-dimensional models where the data is assumed to be generated from an underlying true model $y = f(X) + \varepsilon$ but the user fits a simpler (and wrong) linear model $y= X \beta^0 + \varepsilon$ to the data. Recently, \cite{Li.2020} introduced a bootstrapped version of the desparsified LASSO and showed that for strong signals, i.e., signals for which very few coefficients lie above a certain threshold, this new estimator has a smaller bias as compared to the original desparsified LASSO. A method for variable selection in high-dimensions that controls the directional false discovery rate was proposed in \cite{javanmard2019false} and in \cite{caner2018asymptotically} the desparsified procedure was extended to the conservative LASSO. Furthermore, \cite{ning2017general} introduced the so-called decorrelated score function and established statistical inference procedures for general penalized M-estimators with both convex and nonconvex penalties. In \cite{cheng2019nonparametric}, the desparsified estimator was generalized to nonparametric inference problems such as kernel density estimation and local polynomial regression. In \cite{bellec2022biasing}, a degrees-of-freedom adjustment was proposed that accounts for the dimension of the model selected by the LASSO and established that, when the covariance matrix $\Sigma$ is unknown, a multiplicative correction is necessary for performing statistical inference. The idea of desparsified estimators to vector autoregressive models was extended in \cite{krampe2021bootstrap}. In the present paper we are contributing to the field of uncertainty quantification by extending the desparsified LASSO estimator to measurement matrices beyond the subgaussian case in the almost optimal regime $n \gtrsim s \log^2 p$ provided that s grows sufficiently slow as compared to $p$.

\section{Main theoretical results}\label{sec:main results}
In this section we formally state and prove our main theoretical result, Theorem \ref{thm:mainresult}, which extends the asymptotic normality of the desparsified LASSO from the Gaussian case \cite[Theorem 3.8]{Javanmard.2018} to the case of a random sampling matrix associated to a BOS.
In Section \ref{sec:desparsified_haar} we consider a more general setting of a signal that is not sparse when represented in the canonical basis. In this setting, for the first time, we derive guarantees for the desparsified LASSO in the case that the signal can be sparsely represented by a Haar wavelet transform.

\subsection{Asymptotic normality of desparsified LASSO}
A crucial point for constructing confidence intervals is the knowledge of the distribution of the estimator, the so-called sampling distribution. As it turns out, the desparsified LASSO is approximately normally distributed.

\begin{definition}\cite[Definition 6.12]{wasserman2013all}
An point estimator $\hat{\beta}_n$ for $\beta^0\in\mathbb{R}$ is said to be asymptotically normal if
$$\frac{\hat{\beta}_n-\beta^0}{\sqrt{\var(\hat{\beta}_n)}}\overset{d}{\longrightarrow}\mathcal{N}(0,1),$$
when $n\to\infty$, where $\overset{d}{\longrightarrow}$ means convergence in distribution.
\end{definition}

Let us state our main technical result.

\begin{theorem}\label{thm:mainresult}
Let $\frac{1}{\sqrt{n}} X$ be a normalized random sampling matrix associated to a BOS with constant $K\geq 1$. For $\delta\in(0,1)$ set $t:=(36\cdot \frac{1+\delta}{1-\delta}+1)s_0$ and assume that
$$n\geq CK^2\delta^{-2}t\log(tK^2/\delta)\log(ep),\qquad C>0.$$
Let further $\lambda\geq 2\lambda_0:=2\frac{\sigma\sqrt{K}}{\sqrt{n}}(2+\sqrt{10\log p})$. Then, the following decomposition holds
\begin{equation}\label{eq:key_decomposition}
    \sqrt{n}(\hat{\beta}^u-\beta^0)=W+R,
\end{equation}
where the desparsified LASSO $\hat{\beta}^u$ is defined in \eqref{eq:desparsified_LASSO} with $M=I_{p\times p}$, $W\mid X\sim\mathcal{N}(0,\sigma^2\hat{\Sigma})$ with the noise level $\sigma$ from model \eqref{Cmodel} and, for $\eta \in (0,1)$
\begin{align*}
\mathbb{P}&\left(\Vert R\Vert_{\infty}\leq \frac{8K^{5/2}C_t^\sigma s_0 \sqrt{(\log p)/18}\log(4p/\eta)}{n} + 4 \sqrt{\frac{ K^3 (C_\delta^\sigma)^2 s_0 \log(4p/\eta) \log p}{n}} \right) \\
&\geq 1- \eta -p^{-2} - 2\exp(-\delta^2n/(tK^2)).
\end{align*}
The positive constants $C_t^\sigma$ and $C_\delta^\sigma$ depend only on $\sigma,\delta_t$ and $\sigma, \delta_{s_0}$, respectively.
In particular, if $n\geq C_1 s_0\log^2 p$, then
\begin{equation}\label{eq:main_result}
    \mathbb{P}\left(\Vert R\Vert_{\infty}\geq C_2\frac{\sqrt{s_0}\log p}{\sqrt{n}}\right)\leq 2p^{-2}+2\exp(-\delta^2n/(tK^2))
\end{equation}
where $C_1,C_2> 0$ are constants depending
only on $K, C_\delta^\sigma$ and $C_t^\sigma$.
\end{theorem}

One of the key tools for showing asymptotic normality of de-sparsified estimators is the existence of sharp $\ell_1$ and $\ell_2$ oracle estimates. As it is standard in the compressive sensing literature \cite{Foucart.2013}, we work with the RIP. Since we are interested in a random sampling matrix associated to a BOS, we know from Theorem \ref{thm:rip for BOS} that these matrices fulfill the RIP with high probability provided that the number of measurements $n$ is high enough. In order to establish a bound for the remainder term $R=(M \hat{\Sigma} - I_{p\times p})(\hat{\beta} - \beta^0)$, we start by stating oracle bounds for $\hat{\beta} - \beta^0$, where $\hat{\beta}$ is the LASSO estimator, that minimizes \eqref{eq:LASSO}. In particular, throughout the proof of our main results, we will use the bounds
$$\Vert \hat{\beta}-\beta^0\Vert_2\leq C_{\delta}^{\sigma}\sqrt{K}\cdot \frac{\sqrt{s_0\log p}}{\sqrt{n}} \quad \text{and} \quad \Vert\hat{\beta}-\beta^0\Vert_1\leq C_t^{\sigma}\sqrt{K}\cdot\frac{ s_0\sqrt{\log p}}{\sqrt{n}}$$
with constants $C_{\delta}^{\sigma},C_t^{\sigma}>0$, that hold under the RIP and are proven in the Appendix \ref{subsec:oracle}. Such bounds are widely available in the statistics literature, where it is usually assumed that the design matrix fulfills the restricted eigenvalue condition \cite[Chapter 7]{wainwright2019high} or the compatibility condition \cite[Chapter 6]{Buhlmann.2011}. In contrast, in the literature for matrices associated with a BOS and, more generally, as stated above, in the compressive sensing literature, one usually works with the RIP \cite{Foucart.2013}. In Appendix \ref{appendix} we state and prove, for the sake of completeness, the consistency of the LASSO estimator under the RIP \cite{Foucart.2013}. See \cite{SaraA.vandeGeer.2009, juditsky2011verifiable} for a discussion about the different sufficient conditions for sparse regression and the relationship among them.

We now possess all the tools for proving our main result. For this we decompose the bias term $R$ into two parts, $R^1$ and $R^2$, corresponding to the diagonal and the non-diagonal entries of $\hat{\Sigma}-I_{p\times p}$, respectively. Ideally, we want $\hat{\Sigma}-I_{p\times p}$ to be equal to the zero matrix $0_{p\times p}$ which would require $\langle X_i,X_i\rangle=1$ on the diagonal and $\langle X_i,X_j\rangle=0$ for every off-diagonal entry, $i,j\in[p]$. In order to achieve the aforementioned result in expectation, we will explore the fact that the rows of a matrix associated to a BOS are orthonormal and, consequently, the columns are orthogonal in expectation, c.f. Lemma \ref{le:expectation_sample_cov}.

\begin{proof}[Proof of Theorem \ref{thm:mainresult}]
The proof consists of three major steps: First, we split the data samples in order to make the estimator $\hat{\beta}(y,X)$ independent from $(y,X)$. Next, we prove that the remainder term $R$ is a sum of independent mean-zero bounded random variables. To conclude, we will prove asymptotic bounds for the probability that each of these random variables is large.

In the first part, we show the independence of the matrix rows $x_i$ and the LASSO solution $\hat{\beta}=\hat{\beta}(y,X)$ by using a sample splitting argument. Without loss of generality, we assume that the number of measurements $n$ is even. We use the first half of the data, denoted by $(\tilde{y}_i,\tilde{x}_i)$, to estimate $\hat{\beta}(\tilde{y},\tilde{X})$ while the second half, here denoted by $(y_i,x_i)$, is used to construct the confidence intervals via the desparsified LASSO $\hat{\beta}^u:=\hat{\beta}(\tilde{y},\tilde{X})+\frac{X^*(y-X\hat{\beta}(\tilde{y},\tilde{X}))}{n}.$

By an abuse of notation, during the proof we use the same symbols to denote both subsets of the data, i.e., $(y_i,x_i)$, and we write $n$ instead of $n/2$ for simplicity. Now, we follow the steps of \cite{Javanmard.2018} and decompose $\hat{\beta}^u-\beta^0$ into two terms. Indeed, by writing $y-X\hat{\beta}=\varepsilon+X(\beta^0-\hat{\beta})=\varepsilon+X_j(\beta^0_j-\hat{\beta}_j)+X_{-j}(\beta^0_{-j}-\hat{\beta}_{-j})$, the difference 
$$\sqrt{n}(\hat{\beta}^u-\beta^0)=\sqrt{n}\hat{\beta}+\frac{X^*(y-X\hat{\beta})}{\sqrt{n}}-\sqrt{n}\beta^0$$ becomes componentwise
\begin{alignat*}{2}\label{eq:proofdecomposition}
\sqrt{n}(\hat{\beta}_j^u-\beta^0_j)&=&&\sqrt{n}(1-\frac{1}{n}X_j^*X_j)(\hat{\beta}_j-\beta_j^0)+\frac{1}{\sqrt{n}}X_j^*(\varepsilon+ X_{-j}(\beta^0_{-j}-\hat{\beta}_{-j}))\\
&=&&\sqrt{n}(1-\frac{1}{n}X_j^*X_j)(\hat{\beta}_j-\beta_j^0)+\frac{1}{\sqrt{n}}X_j^*X_{-j}(\beta^0_{-j}-\hat{\beta}_{-j})+\frac{1}{\sqrt{n}}X_j^*\varepsilon.
\end{alignat*}
Defining $W_j:=\frac{1}{\sqrt{n}}X_j^*\varepsilon$ and conditioning on $X$, we obtain that $W\mid X\sim\mathcal{N}(0,\sigma^2\hat{\Sigma})$, since $\varepsilon\sim\mathcal{CN}(0,\sigma^2 I_{n\times n})$. Furthermore, by defining $u:=\hat{\beta}-\beta^0$, and conditioning on $u$, the term
\begin{align*}
    R^{(1)}_j:=\sqrt{n}(1-\frac{1}{n}X_j^*X_j)u_j=\frac{1}{\sqrt{n}}\sum\limits_{i=1}^n(1-\vert X_{ij}\vert^2)u_j
\end{align*}
consists of mean-zero bounded random variables $Z_i^{(1)}:=(1-\vert X_{ij}\vert^2)u_j$. Due to the sample splitting $u$ is independent of $X$. With an abuse of notation, whenever we write for example $\mathbb{E}[Z_i^{(1)}]$, we actually mean $\mathbb{E}[Z_i^{(1)} | u]$. Indeed, since $\mathbb{E}[x_ix_i^*]=I_{p\times p}$, the mean of $Z_i^{(1)}$ is given by
\begin{align*}
    \mathbb{E}[Z_i^{(1)}]&=(1-\mathbb{E}[\overline{X}_{ij}X_{ij}])u_j=(1-e_j^T\mathbb{E}[x_ix_i^*]e_j)u_j=0.
\end{align*}
Using $\vert X_{ij}\vert\leq K$, the variance of $Z_i^{(1)}$ can be estimated by
\begin{align*}
    \mathbb{E}[\vert Z_i^{(1)}\vert^2] &=\vert u_j\vert^2\mathbb{E}[(1-\overline{X}_{ij}X_{ij})(1-\overline{X}_{ij}X_{ij})]
    =\vert u_j\vert^2(1-2+\mathbb{E}[\overline{X}_{ij}\vert X_{ij}\vert^2X_{ij}])\\
    &=\vert u_j\vert^2(\mathbb{E}[\vert X_{ij}\vert^2\vert X_{ij}\vert^2-1])
    \leq (K^2-1)\vert u_j\vert^2.
\end{align*}
Finally,
\begin{align*}
    \vert Z_j^{(1)}\vert&=\vert 1-\vert X_{ij}\vert^2\vert\cdot\vert u_j\vert\leq K^2\Vert u\Vert_{\infty}\leq K^2\Vert u\Vert_{1}.
\end{align*}

Secondly, we rewrite the remaining error term as
\begin{align*}
R_j^{(2)}&=\frac{1}{\sqrt{n}}X_j^*X_{-j}u_{-j}=\frac{1}{\sqrt{n}}\sum\limits_{i=1}^nX_{ji}^*(X_{-j}u_{-j})_i\\
&=\frac{1}{\sqrt{n}}\sum\limits_{i=1}^n\overline{X}_{ij}\langle (\overline{x}_i)_{-j},u_{-j}\rangle=\frac{1}{\sqrt{n}}\sum\limits_{i=1}^n\overline{X}_{ij}\langle \overline{x}_i,u^{(-j)}\rangle,
\end{align*}
where $u^{(-j)}:=u-u_j e_j$ is the vector $u$ with its $j$-th component set to zero. With the same argument as above, for fixed $j$ and conditioned on $u$ the random variables $Z_i^{(2)}:=\overline{X}_{ij}\langle \overline{x}_i,u^{(-j)}\rangle$ are independent, mean-zero and bounded. Indeed, the expected value of $Z_i^{(2)}$ is given by
$$\mathbb{E}[Z_i^{(2)}]=\mathbb{E}[\overline{X}_{ij}\overline{x}_i^*]u^{(-j)}=e_j^T\mathbb{E}[\overline{x}_i\overline{x}_i^*]u^{(-j)}=e_j^Tu^{(-j)}=0$$
and the variance can be estimated by
\begin{align*}
    \mathbb{E}[\vert Z_i^{(2)}\vert^2]=(u^{(-j)})^*\mathbb{E}[\overline{x}_i\overline{X}_{ij}^*\overline{X}_{ij}\overline{x}_i^*]u^{(-j)}=(u^{(-j)})^*\mathbb{E}[\overline{x}_i\vert \overline{X}_{ij}\vert^2\overline{x}_i^*]u^{(-j)}\leq K^2 \Vert u\Vert_2^2.
\end{align*}
Furthermore, we have the bound 
$$\vert Z_i^{(2)}\vert\leq K\vert \langle \overline{x}_i,u^{(-j)}\rangle\vert\leq K^2\Vert u\Vert_1.$$
After introducing $Z_i:=Z_i^{(1)}+Z_i^{(2)}$ we set $R=R^1+R^2=\frac{1}{\sqrt{n}}\sum\limits_{i=1}^nZ_i^{(1)}+Z_i^{(2)}=\frac{1}{\sqrt{n}}\sum\limits_{i=1}^nZ_i$. Hence, 
$$\vert Z_i\vert^2\leq \vert Z_i^{(1)}\vert^2+2\vert Z_i^{(1)}\vert\cdot \vert Z_i^{(2)}\vert+\vert Z_i^{(2)}\vert^2\leq 2(\vert Z_i^{(1)}\vert^2+\vert Z_i^{(2)}\vert^2).$$
Therefore, conditioned on the events $\mathcal{F}:=\{\varepsilon\in\mathbb{C}^n:\max\limits_{j\in[p]}\frac{2}{n}\vert\langle\varepsilon,X_j\rangle\vert\leq\lambda_0\}$ and $\mathcal{G}:=\{\frac{1}{\sqrt{n}}X\in\mathbb{C}^{n\times p}:\frac{1}{\sqrt{n}}X\text{ satisfies RIP of order } t\text{ with constant }\delta_t<\delta\}$, we obtain by Lemma \ref{le:l2oracleinequality} and Theorem \ref{thm:LASSOforRIP}, respectively, that
\begin{align*}
    &\mathbb{E}[Z_i]=\mathbb{E}[Z_i^{(1)}]+\mathbb{E}[Z_i^{(2)}]=0,\\
    &\mathbb{E}[\vert Z_i\vert^2]\leq 2(\mathbb{E}\vert Z_i^{(1)}\vert^2+\mathbb{E}\vert Z_i^{(2)}\vert^2)\leq 4K^2\Vert u\Vert_2^2\leq 4K^3\left(C_{\delta}^{\sigma}\right)^2\frac{s_0\log p}{n},\\
    &\vert Z_i\vert\leq \vert Z_i^{(1)}\vert+\vert Z_i^{(2)}\vert\leq 2K^2\Vert u\Vert_1\leq 8K^2\frac{\lambda_0s_0}{1-\delta_t}\leq 2K^{5/2} C_t^{\sigma} \frac{s_0\sqrt{\log p}}{\sqrt{n}},
\end{align*}
with $C_t^{\sigma}=\frac{16\sqrt{10}\sigma}{1-\delta_t}$, where we also used the choice $\lambda_0=\frac{\sigma\sqrt{K}}{\sqrt{n}}(2+\sqrt{10\log p})$ in Lemma \ref{6.2}. This, in particular, means that the event $\mathcal{F}$ will occur with high probability. As for the event $\mathcal{G}$, it will also occur with high probability since we have assumed $n\geq CK^2\delta^{-2}t\log(tK^2/\delta)\log(ep)$, cf. Theorem \ref{thm:rip for BOS}. Now, we fix $j\in[p]$ and decompose $\sum\limits_{i=1}^nZ_i$ into its real and imaginary parts and bound each one of them separately,
\begin{align*}
\mathbb{P}(\vert R_j\vert\geq t)&=\mathbb{P}\left(\frac{1}{\sqrt{n}}\left\vert\sum\limits_{i=1}^nZ_i\right\vert\geq t\right)
=\mathbb{P}\left(\left(\Re\sum\limits_{i=1}^nZ_i\right)^2+\left(\Im\sum\limits_{i=1}^nZ_i\right)^2\geq t^2n\right)\\
&\leq \mathbb{P}\left(\left(\Re\sum\limits_{i=1}^nZ_i\right)^2\geq\frac{t^2n}{2}\right)+\mathbb{P}\left(\left(\Im\sum\limits_{i=1}^nZ_i\right)^2\geq\frac{t^2n}{2}\right)\\
&=\mathbb{P}\left(\left\vert\sum\limits_{i=1}^n\Re Z_i\right\vert\geq \frac{t\sqrt{n}}{\sqrt{2}}\right)+\mathbb{P}\left(\left\vert\sum\limits_{i=1}^n\Im Z_i\right\vert\geq \frac{t\sqrt{n}}{\sqrt{2}}\right),
\end{align*}
where the probabilities above are conditioned on $u:=\hat{\beta}-\beta^0$.
Due to the sample splitting introduced at the beginning of the proof, the variables $Z_i$ are independent. Therefore, Bernstein's inequality (Theorem \ref{bernstein}) yields 
$$\mathbb{P}\left(\frac{1}{\sqrt{n}}\left\vert\sum\limits_{i=1}^nZ_i\right\vert\geq t\right)\leq 4\exp\left(-\frac{t^2n/4}{4K^3\left(C_{\delta}^{\sigma}\right)^2s_0\log p+2K^{5/2}C_t^\sigma s_0\sqrt{\log p}t/\sqrt{18}}\right).$$
Taking the union bound over all entries $R_j$, $j=1,\hdots,p$ gives
\begin{align*}
\mathbb{P}(\Vert R\Vert_{\infty}\geq t) & =\mathbb{P}\left(\max\limits_{j\in[p]}\vert R_j\vert\geq t\right)\leq\sum\limits_{j=1}^p\mathbb{P}(\vert R_j\vert\geq t) \\
& \leq
4p \exp\left(-\frac{t^2n/4}{4 K^3 \left(C_{\delta}^{\sigma}\right)^2s_0\log p+2K^{5/2} C_t^\sigma s_0\sqrt{\log p}t/\sqrt{18}}\right).
\end{align*}
The left hand side is at most $\eta \in (0,1)$ if
\[
\frac{t^2n/4}{4 K^3 (C_{\delta}^\sigma)^2 s_0 \log p + 2K^{5/2} C_{t}^\sigma s_0 \sqrt{(\log p)/18}\; t} \geq \log(4p/\eta),
\]
which is equivalent to
\[
a t^2 - 2b t -c \geq 0
\]
with $a = n/(4 \log(4p/\eta))$, $b = K^{5/2} C_t^\sigma s_0 \sqrt{(\log p)/18}$ and
$c = 4 K^3 (C_\delta^\sigma)^2 s_0 \log p$.
The above inequality is equivalent to 
$t \geq \frac{b}{a} + \sqrt{\frac{c}{a} + \frac{b^2}{a^2}}$
which in turn is implied by
\[
t \geq \frac{2b}{a} + \sqrt{\frac{c}{a}}.
\]
Setting $t$ equal to the right hand side above and plugging in the values of $a,b,c$ we obtain
\[
\mathbb{P}\left(\Vert R\Vert_{\infty}\leq \frac{8K^{5/2}C_t^\sigma s_0 \sqrt{(\log p)/18}\log(4p/\eta)}{n} + 4 \sqrt{\frac{K^3 (C_\delta^\sigma)^2 s_0 \log(4p/\eta) \log p}{n}} \right) \geq 1- \eta.
\]
Setting $\eta = p^{-c}$ and assuming the regime $n \geq K^3 (C_\delta^\sigma)^2 s_0 \log(4p/\eta) \log(p)$ this implies (by noting that $x \leq \sqrt{x}$ for $x \in [0,1]$)
\[
\mathbb{P}\left(\Vert R\Vert_{\infty}\leq
\underbrace{\left(4 + \frac{8C_t^\sigma}{\sqrt{K}(C_\delta^\sigma)^2 \sqrt{18\log p}}\right) \sqrt{\frac{(K^2 C_\delta^\sigma)^2 s_0 \log(4p^{c+1}) \log p}{n}}}_{=:q} \right) \geq 1-p^{-c}.
\]

By the law of total probability, Lemma \ref{6.2} and Theorem \ref{thm:rip for BOS} we obtain
\begin{align*}
    \mathbb{P}(\Vert R\Vert_{\infty}\geq q) &= \mathbb{P}(\Vert R\Vert_{\infty}\geq q\mid \mathcal{F}\cap \mathcal{G})\mathbb{P}(\mathcal{F}\cap\mathcal{G})+\mathbb{P}(\Vert R\Vert_{\infty}\geq q\mid (\mathcal{F}\cap\mathcal{G})^c)\cdot\mathbb{P}((\mathcal{F}\cap\mathcal{G})^c)\\
    &\leq \mathbb{P}(\Vert R\Vert_{\infty}\geq q\mid \mathcal{F}\cap \mathcal{G})+\mathbb{P}({\mathcal{F}^c})+\mathbb{P}({\mathcal{G}^c})\\
    &\leq p^{-c}+p^{-2}+2\exp(-\delta^2n/(tK^2)).
\end{align*}
Choosing $c=2$ yields $ \mathbb{P}(\Vert R\Vert_{\infty}\geq q)\leq 2p^{-2}+2\exp(-\delta^2n/(tK^2))$. This ends the proof.
\end{proof}

In order to guarantee an approximate Gaussian distribution for $\sqrt{n}(\hat{\beta}^u-\beta^0)$, two sufficient conditions play a role - the fact that the measurement matrix satisfies the RIP, which requires $n \gtrsim  s_0\log^2 s_0\log p$ samples and the fact that the bias term $R$ asymptotically vanishes if $n \gtrsim  s_0\log^2 p $. Therefore, in very precise terms, our sample complexity reads as $ n \gtrsim \max\{ s_0\log^2 s_0\log p, s_0\log^2 p \}$.

The proof of \cite[Theorem 3.8]{Javanmard.2018} avoids sample splitting by using a refined analysis with a leave-one-out approach. A crucial point is the independence of $X\Sigma^{-1}e_i$ and $X_{-i}$, which holds for a Gaussian design with independently sampled rows  \cite[Lemma 3.6]{Javanmard.2018}, but not for a (subsampled) Fourier matrix. Therefore, we believe that it is a challenging task to work without sample splitting in our case.

\section{Confidence regions}\label{sec:confidence regions}

Theorem \ref{thm:mainresult} established the asymptotic normality for the c-LASSO, i.e., the normality of the differences $\sqrt{n}(\hat{\beta}^u-\beta^0)$ and $\sqrt{n}(\hat{\beta}^u_g-\beta^0)$ in the case where $n,p\to\infty$ and $p/n\to\kappa \in (0, \infty)$ provided that $n \gtrsim s_0^2 \log p$ and $n \gtrsim s_g^2 \log p$, respectively. This allows us to perform uncertainty quantification for $\beta^0$ based on the estimator $\hat{\beta}^u$ (or $\hat{\beta}_g^u$). Given the way how c-LASSO connects the real and imaginary parts the natural shape for a confidence region will be a circle. Indeed, the set $J^{\circ}_i(\alpha):=\{z\in\mathbb{C}:\vert\hat{\beta}_i^u-z\vert\leq\delta_i^{\circ}(\alpha)\}$ with length
$$\delta_i^{\circ}(\alpha):=\frac{\sigma\hat{\Sigma}_{ii}^{1/2}}{\sqrt{2n}}\Psi(1-\alpha)^{-1},$$
contains the ground truth $\beta^0$ with a significance level $\alpha\in (0,1)$. We denote by $\Psi(\,\cdot\,)^{-1}$ the quantile function of a Rayleigh distribution with parameter equal to $1$, which can be analytically expressed as $\Psi(1-\alpha)^{-1}=\sqrt{2\log{(1/\alpha)}}$. Note, however, that the confidence region depends on the true noise standard deviation $\sigma$ that is not known a priori. Nevertheless, the confidence regions will still be valid if one replaces $\sigma$ with a consistent estimator $\hat{\sigma}$. A discussion for valid noise estimators will be found at the end of this section. See \cite[Chapter 6]{wasserman2013all} for a discussion on the background of pointwise asymptotic confidence intervals. Our main UQ result for the c-LASSO is given by

\begin{theorem}\label{thm:confidence circle}
Let $\frac{1}{\sqrt{n}} X$ be a normalized random sampling matrix associated to a BOS with constant $K\geq 1$. For $\delta\in(0,1)$ set $t:=(36\cdot \frac{1+\delta}{1-\delta}+1)$ and assume that $n\geq CK^2\delta^{-2}t\log(tK^2/\delta)\log(ep)$ with a constant $C>0$. Let $\lambda\geq 2\lambda_0$. Let further $\hat{\sigma}$ be a noise estimator satisfying \eqref{eq:consistent_noise}.
Assume an asymptotic regime where $n\to\infty$ and $p=p(n)\to\infty$ with $p/n\to\kappa \in (0, \infty)$ such that $\frac{s_0\log^2 p}{n}\to 0$. For every $i\in [p]$ set the confidence circle for $\beta_i^0$ estimated via the desparsified LASSO derived from c-LASSO as
\begin{equation}\label{eq:def_conf_circle}
    J^{\circ}_i(\alpha):=\{z\in\mathbb{C}:\vert\hat{\beta}_i^u-z\vert\leq\delta_i^{\circ}(\alpha)\},\qquad \alpha\in(0,1),
\end{equation}
with radius $\delta^{\circ}(\alpha):=\frac{\hat{\sigma}\hat{\Sigma}_{ii}^{1/2}}{\sqrt{n}}\sqrt{\log\left(\frac{1}{\alpha}\right)}$. Then, $\lim\limits_{n\to\infty}\mathbb{P}\left(\beta_i^0\in J^{\circ}_i(\alpha)\right)=1-\alpha.$
\begin{proof}
First, we note that
$$\lim\limits_{n\to\infty}\mathbb{P}\left(\vert\hat{\beta}^u_i-\beta_i^0\vert\leq\delta^{\circ}(\alpha)\right)=1-\lim\limits_{n\to\infty}\mathbb{P}\left(\vert\hat{\beta}^u_i-\beta_i^0\vert >\delta^{\circ}(\alpha)\right).$$
By Lemma \ref{C13Re} below, we conclude that for fixed $i$ the real random vector
$$\tilde{g}_i:=\frac{g_i}{\hat{\sigma}/\sqrt{2n}\hat{\Sigma}_{ii}^{1/2}}:=\begin{pmatrix}
\frac{\sqrt{n}\Re(\hat{\beta}_i^u-\beta_i^0)}{\hat{\sigma}/\sqrt{2}\hat{\Sigma}_{ii}^{1/2}} \\
\frac{\sqrt{n}\Im(\hat{\beta}_i^u-\beta_i^0)}{\hat{\sigma}/\sqrt{2}\hat{\Sigma}_{ii}^{1/2}}
\end{pmatrix}$$
is asymptotically standard Gaussian distributed:
$$\lim\limits_{n\to\infty}\sup\limits_{\Vert\beta_0\Vert_0\leq s_0}\tilde{g}_i\sim\mathcal{N}(0,I_{2\times 2}).$$
Hence, by taking $c=\sqrt{2\log{(1/\alpha)}}$ and applying the dominated convergence theorem, we have
\begin{align*}
\lim\limits_{n\to\infty}\mathbb{P}\left(\vert\hat{\beta}^u_i-\beta_i^0\vert > \delta^{\circ}(\alpha)\right)&=\lim\limits_{n\to\infty}\mathbb{P}\left(\Vert g_i\Vert_2 > \delta^{\circ}(\alpha)\right)=\lim\limits_{n\to\infty}\mathbb{P}\left(\Vert \tilde{g}_i\Vert_2 > c\right)\\
&=\mathbb{P}\left(\Vert\lim\limits_{n\to\infty} \tilde{g}_i\Vert_2 > c\right)=\mathbb{P}\left(\Vert\xi\Vert_2 > c\right),\\
\end{align*}
where $\xi\sim\mathcal{N}(0,I_{2\times 2})$. 
Therefore, $\Vert \xi\Vert_2$ is $\chi$ distributed with two degrees of freedom, also known as the Rayleigh distribution \cite[Chapter 11.3]{evans2011statistical}. In contrast to the one-dimensional standard Gaussian distribution, there is an analytic expression for the quantile, given by
$\mathbb{P}(\Vert \xi\Vert_2 > c)=\int\limits_{c}^{\infty}x\cdot e^{\frac{-x^2}{2}}dx=e^{-\frac{c^2}{2}}=\alpha.$ This concludes the proof.
\end{proof}
\end{theorem}

It is important to note that the ground truth $\beta_i^0$ is fixed and the confidence interval is random as $\hat{\beta}^u_i$ is random. Having realized that, we interpret \eqref{eq:def_conf_circle} as follows: We conduct an experiment a certain number of times to retrieve the ground truth which is fixed based on the model \eqref{Cmodel}. From the practitioner's point of view, in every experiment we collect data. This data is different due to a different realization of the noise, which is random. However, the measurement matrix $X$ is fixed primarily at the beginning of the experiment. This is exactly the setting in Theorem \ref{thm:mainresult} where we proved $W$ to be Gaussian distributed under the condition that $X$ is given. With the different data, we construct each time a different confidence interval. Finally, \eqref{Cmodel} assures that in a fraction of $(1-\alpha)$ of the cases we construct the confidence intervals in a way that they contain the ground truth.

\begin{remark}[Confidence intervals]\label{re:1d_conf_int}
In addition to the confidence circle, under the assumptions of Theorem \ref{thm:confidence circle}, we are able to construct $(1-\alpha)$ confidence intervals for the real and imaginary part of $\beta_i^0$ separately, i.e.
$$J^{\Re}_i(\alpha):=[\Re(\hat{\beta}^u_i)-\delta(\alpha),\Re(\hat{\beta}^u_i)+\delta(\alpha)] \
\text{and} \
J^{\Im}_i(\alpha):=[\Im(\hat{\beta}^u_i)-\delta(\alpha),\Im(\hat{\beta}^u_i)+\delta(\alpha)],$$
with radius given by
\begin{equation}\label{eq:conf_int_radius}
    \delta(\alpha):=\frac{\hat{\sigma}\cdot\hat{\Sigma}_{ii}^{1/2}}{\sqrt{2n}}\Phi^{-1}(1-\alpha/2),
\end{equation}
where $\Phi^{-1}$ is the quantile function of the standard normal distribution. The confidence intervals $J_i^{\Re}(\alpha)$ and $J_i^{\Im}(\alpha)$ are asymptotically valid for the real and imaginary part, respectively, i.e.,
\begin{equation}\label{eq:confidence interval}
    \lim\limits_{n\to\infty}\mathbb{P}\left(\Re(\beta^0_i)\in J_i^{\Re}(\alpha)\right)=1-\alpha
\end{equation}
and
$$\lim\limits_{n\to\infty}\mathbb{P}\left(\Im(\beta^0_i)\in J_i^{\Im}(\alpha)\right)=1-\alpha.$$
The proof is based on \cite[Theorem 15]{Javanmard.2014} and presented in Appendix \ref{subsec:proof_conf_int_valid}.
\end{remark}

Theorem \ref{thm:confidence circle} establishes honest confidence intervals of length $\frac{1}{\sqrt{n}}$ and the natural question concerns the optimality of this result. The work in \cite{Cai.2017} established minimax guarantees for confidence intervals for high-dimensional problems, later generalized in \cite{Javanmard.2018}, for the case of Gaussian designs. In order to explain these results, we define the set of all $(1-\alpha)$ confidence intervals for $\beta_i^0$ with fixed $i\in[p]$ as
$$\mathclap{I}_{\alpha}(\Theta)=\{J_i^{\Re}(\alpha,Z)=[l(Z),u(Z)]:\inf\limits_{\beta_i^0\in\Theta}\mathbb{P}(l(Z)\leq\beta_i^0\leq u(Z))\geq 1-\alpha\},$$
where $Z=(\Re(y_1,x_1),...,\Re(y_n,x_n))^T\in\mathbb{R}^{n\times p+1}$ is the given data and $\mathbb{P}$ is the induced probability distribution on $(y, X)$ for the design matrix $X$ and noise realization $\varepsilon$, given the fixed signal $\beta^0$. The maximal expected length of an $(1-\alpha)$-confidence interval is
$L(J_i^{\Re}(\alpha,Z),\Theta)=\sup\limits_{\beta_i^0\in\Theta}\mathbb{E}[L(J_i^{\Re}(\alpha,Z))]$ where the length is given by $L(J_i^{\Re}(\alpha,Z))=u(Z)-l(Z)$. Finally, the minimax expected length is defined as
$$L_{\alpha}^*(\Theta)=\inf\limits_{J_i^{\Re}(\alpha,Z)\in \mathclap{I}_{\alpha}}L(J_i^{\Re}(\alpha,Z),\Theta).$$
For a known sparsity level $s_0$, \cite[Theorem 1]{Cai.2017} states that for Gaussian designs, the minimax expected length of a $(1-\alpha)$ confidence intervals is of order $\frac{1}{\sqrt{n}}+\frac{s_0\log p}{n}$. Later, by considering more structure on the inverse of the covariance matrix, \cite[Proposition 4.2]{Javanmard.2018} generalized the lower bound in the previous result and established that the minimax expected length for $(1 - \alpha)$-confidence intervals of $\beta_i$ has length at least 
$$ L_{\alpha}^*(\Theta(s_0, s_{\Omega}, \rho)) \geq \frac{1}{\sqrt{n}} + \min \left\{ s_0 \frac{\log p}{n}, s_{\Omega} \frac{\log p}{n}, \rho \sqrt{\frac{\log p}{n}}\right\},$$
where $ \Theta(s_0, s_{\Omega}, \rho)$ is the set of parameters $(\beta^0, \Sigma^{-1}, \sigma^2)$ given by

\begin{align*}
   \Theta(s_0, s_{\Omega}, \rho) &= \Big\{ \gamma =(\beta^0, \Sigma^{-1}, \sigma^2): ||\beta^0||_0 \leq s_0, \sigma^2 \in (0, c], (\Sigma)_{ii}\leq 1, \\& \frac{1}{C_{\max}} < \sigma_{\min}(\Sigma^{-1}) < \sigma_{\max}(\Sigma^{-1}) < \frac{1}{C_{\min}},\\
   &|| \Sigma^{-1}||_{\infty} \leq \rho, \max_{i \in [p]} | \{ j \neq i, (\Sigma^{-1})_{ij} \neq 0 \} | \leq s_{\Omega}\Big\} .
\end{align*}

Note that although the scenario considered here in this paper (BOS design matrices and Gaussian noise) is compatible with the set $ \Theta(s_0, s_{\Omega}, \rho)$, the theorems above were proven for Gaussian designs. Still, one should not expect a better lower bound for matrices associated with a BOS as the one compared with Gaussian designs, which indicates the optimality of length of the constructed confidence intervals with radius according to \eqref{eq:conf_int_radius}.

Moreover, in the particular case of subsampled Fourier matrices as defined in Definition \ref{def:partial Fourier} the size of the confidence region's radii does not depend on the component $i\in[p]$. For the entries of the sample covariance we obtain
\begin{align*}
\hat{\Sigma}_{lj}&=\frac{1}{n}\sum\limits_{k=1}^n\overline{(F_{\Omega})}_{kl}(F_{\Omega})_{kj}
=\frac{1}{n}\sum\limits_{k=1}^n\overline{F_{t_kl}}F_{t_kj}=\frac{1}{n}\sum\limits_{k=1}^ne^{\frac{2\pi i(t_k-1)(j-l)}{p}}.
\end{align*}
On the diagonal $\hat{\Sigma}_{ii}^{1/2}=1$ and hence the confidence regions have the same size for every component of $\beta^0$.

In the remainder of this chapter, we turn to the topic of noise estimation. From a mathematical point of view, estimating the error variance for high-dimensional estimators is a non-trivial problem that still attracts significant interest. In the LASSO case, the works \cite{reid2016study, yu2019estimating, giraud2012high} discuss the noise level estimation problem.
The most common method used in the desparsified LASSO literature, e.g., in \cite{Zhang.2014, bellec2022biasing, Javanmard.2018, Li.2020, vandeGeer.2014} to name just a few, for estimating the noise level is the scaled LASSO \cite{sun2012scaled}. Other alternative methods, still to be explored, are for example presented in \cite{dicker2014variance, liu2020estimation, kennedy2020greedy}.

In the MRI setting the noise can be measured directly during MR image acquisition, for instance, during the so-called pre-scan procedure, which is mandatory for every patient, yielding a direct, ground truth estimation of the noise \cite{Leussler.10162015, Biber.1052017}. Alternatively, it may be estimated retrospectively (indirectly) from the final image, in case the directly determined noise estimation is not accessible anymore. For a review of different noise estimation methods in MRI see, e.g., \cite{AjaFernandez.2016}.

Lemma 13 in \cite{Javanmard.2014} shows, that the asymptotic normality still holds when the true noise level is replaced by a consistent noise estimator in the sense of \eqref{eq:consistent_noise}. Here we provide a complex version of this lemma. The proof is presented in Appendix \ref{subsec:proof_of_noise_consistency}.

\begin{lemma}\label{C13Re}
Let $\frac{1}{\sqrt{n}} X$ be a normalized random sampling matrix associated to a BOS with constant $K\geq 1$. Let $\lambda\geq 2\lambda_0$. Let further $\hat{\sigma}=\hat{\sigma}(y,X)$ be a noise estimator with
\begin{equation}\label{eq:consistent_noise}
\lim\limits_{n\to\infty}\sup\limits_{\beta^0\in\mathbb{C}^p:\Vert\beta^0\Vert_0\leq s_0}\mathbb{P}\left(\left\vert\frac{\hat{\sigma}}{\sigma}-1\right\vert\geq \varepsilon\right)=0\quad\forall \varepsilon>0.
\end{equation}
If $n\to\infty$ and $p=p(n)\to\infty$ with $p/n\to\kappa \in (0, \infty)$ such that $\frac{s_0\log^2 p}{n}\to 0$, then, for all $x\in\mathbb{R}$, we have
\begin{equation*}
\lim\limits_{n\to\infty}\sup\limits_{\beta^0\in\mathbb{C}^p:\Vert\beta^0\Vert_0\leq s_0}\left\vert\mathbb{P}\left(\frac{\sqrt{n}\Re(\hat{\beta}_i^u-\beta_i^0)}{\hat{\sigma}/\sqrt{2}\hat{\Sigma}_{ii}^{1/2}}\leq x\right)-\Phi(x)\right\vert=0
\end{equation*}
and
\begin{equation*}
\lim\limits_{n\to\infty}\sup\limits_{\beta^0\in\mathbb{C}^p:\Vert\beta^0\Vert_0\leq s_0}\left\vert\mathbb{P}\left(\frac{\sqrt{n}\Im(\hat{\beta}_i^u-\beta_i^0)}{\hat{\sigma}/\sqrt{2}\hat{\Sigma}_{ii}^{1/2}}\leq x\right)-\Phi(x)\right\vert=0.
\end{equation*}
\end{lemma}

\section{Desparsified LASSO for vectors that are sparse under a basis transform}\label{sec:desparsified_haar}

As described in the introduction, one of the goals of UQ for high-dimensional image problems is to obtain confidence intervals for every pixel of an image retrieved as a solution to a high-dimensional problem. In the case where the image is not sparse in the canonical basis, it should be first sparsified via, for example, an orthogonal transform, before the desparsified estimator can be applied. This section extends the desparsified LASSO to the case where the ground truth $\beta^0$ is sparse in the Haar wavelet domain. See \cite[Chapter 6]{walnut2002introduction} for more details about the Haar wavelet transform.

\begin{theorem}\label{thm:mainresult_sparse_transform}
Let $\frac{1}{\sqrt{n}} F_{\Omega}$ be a normalized subsampled Fourier matrix and $H\in\mathbb{C}^{p\times p}$ the orthogonal Haar wavelet transform. Set $\kappa_j=\frac{3\sqrt{2\pi}}{\sqrt{j}}$ for $j\in[p]$. Assume that the transformed ground truth $z^0=H\beta^0$ is $s_0$-sparse. For $\delta\in (0,1)$ set $t:=(36\cdot \frac{1+\delta}{1-\delta}+1)s_0$ with $t\geq \log p$, assume
$$n\geq C\delta^{-2}\Vert \kappa\Vert_2^2 t \log^2 t \log p,$$
and choose $n$ (possibly not distinct) indices $j\in\Omega\subset [p]$ i.i.d. from the probability measure $\nu$ on $[p]$ given by
$$\nu(j)=\frac{\kappa_j^2}{\Vert \kappa\Vert_2^2}.$$
Let $D=\operatorname{diag}(d)\in\mathbb{C}^{p\times p}$ be the diagonal matrix with $d_j=\Vert \kappa\Vert_2/(\sqrt{p}\kappa_j)$. Let $\hat{z}$ be the LASSO solution of
\begin{equation*}\label{eq:LASSO_z}
 \hat{z} = \argmin\limits_{z\in\mathbb{C}^p}\frac{1}{2n}\Vert D(F_{\Omega}H^*z-y)\Vert_2^2+\lambda\Vert \beta\Vert_1
\end{equation*}
and the desparsified LASSO $\hat{\beta}^u$ is defined via
\begin{equation}\label{eq:desparsified_LASSO_sparsetrafo}
    \hat{\beta}^u=H^*\hat{z}+\frac{1}{n}F_{\Omega}^*D^2(y-F_{\Omega}H^*z^0).
\end{equation}
Let $\sigma$ be the noise level from model \eqref{Cmodel}. Then, the following decomposition holds
$$\sqrt{n}(\hat{\beta}^u-\beta^0)=W+R,$$
where 
$W\mid D^2F_{\Omega}\sim\mathcal{CN}\left(0,\frac{\sigma^2}{n}F_{\Omega}^*D^4F_{\Omega}\right)$ with variance $\frac{\sigma^2}{n}(F_{\Omega}^*D^4F_{\Omega})_{ii}=\frac{\sigma^2}{n}\Vert D^2\Vert_F^2$ for every $i\in[p]$ and, if $n\geq C_1 s_0\log p$, then
\begin{equation}\label{eq:main_result_sparse_trafo}
    \mathbb{P}\left(\Vert R\Vert_{\infty}\geq C(\Vert \kappa\Vert_2,\sigma,\delta_t)\frac{\sqrt{s_0}\log p}{\sqrt{n}}\right)\leq 2p^{-2}+p^{-c\log^3 t},
\end{equation}
where $C(\Vert\kappa\Vert_2,\sigma,\delta_t)>0$ is a constant depending on $\Vert\kappa\Vert_2$, $\sigma$ and the RIP constant $\delta_t$.
\end{theorem}

The construction of the confidence regions follows the same procedure as in Section \ref{sec:confidence regions}. The major difference is the change in the sampling pattern. In Theorem \ref{thm:mainresult} the rows of the design matrix are selected uniformly on $[p]$, whereas here the rows are sampled with respect to the non-uniform probability measure $\nu$. The desparsified LASSO $\hat{\beta}^u$, as defined in \eqref{eq:desparsified_LASSO_sparsetrafo}, is asymptotically normal with covariance matrix $\frac{\sigma^2}{n}F_{\Omega}^*D^4F_{\Omega}$. The sample size remains the same as in Theorem \ref{thm:mainresult} except for additional $\sqrt{\log_2 p}$-factors as described in the following remark.

\begin{remark}\label{re:BOS_sparse_trafo}
    One main ingredient for the following proof is that the measurement matrix is associated to a BOS. More precisely, under the stated assumptions, the proof of Theorem \ref{thm:sparse_transform_RIP} (\cite[Theorem 1]{6651836}) assures that the matrix $A:=DF_{\Omega}H^*$ is with high probability associated to a BOS with constant $\Vert \kappa\Vert_2$. Theorem 4 in \cite{6651836} provides an upper bound for this constant $\Vert \kappa \Vert_2\leq 52\sqrt{\log_2 p}$.
\end{remark}

\begin{proof}[Proof of Theorem \ref{thm:mainresult_sparse_transform}]
The underlying model
$$Dy=DF_{\Omega}\beta^0+D\varepsilon=DF_{\Omega}H^*z^0+D\varepsilon,$$
is like the standard inverse problem with the difference that $F_{\Omega}$ is replaced with $DF_{\Omega}$. The desparsified LASSO for $z^0$ reads as
\begin{equation*}
    \hat{z}^u=\hat{z}+\frac{1}{n}(DF_{\Omega}H^*)^*(Dy-DF_{\Omega}H^*\hat{z}),
\end{equation*}
where $\hat{z}$ is the LASSO estimator for $z^0$. Due to Remark \ref{re:BOS_sparse_trafo} the matrix $A:=DF_{\Omega}H^*$ is (with high probability - which is considered at the end of the proof) a random sampling matrix associated to a BOS with constant $\Vert\kappa\Vert_2$, and hence the matrix $M$ in \eqref{eq:desparsified_LASSO} is chosen to be the identity $I_{p\times p}$.
Defining $\hat{\Sigma}_A:=\frac{1}{n}A^*A$ the estimator $H^*\hat{z}^u$ for $H^*z^0=\beta^0$ fulfills the decomposition:
\begin{align*}
    \sqrt{n}(H^*\hat{z}^u-\beta^0)&=\frac{1}{\sqrt{n}}H^*(DF_{\Omega}H^*)^*D\varepsilon+\sqrt{n}H^*(\hat{\Sigma}_AH H^*-I_{p\times p})(z^0-\hat{z})\\
    &=\underbrace{\frac{1}{\sqrt{n}}(D^2F_{\Omega})^*\varepsilon}_{:=W}+\underbrace{\sqrt{n}(H^*\hat{\Sigma}_AH-I_{p\times p})(\beta^0-H^*\hat{z})}_{:=R}
\end{align*}
Conditioned on $(D^2F_{\Omega})$, the term $W$ is Gaussian distributed with covariance matrix $(\sigma^2/n)\cdot F_{\Omega}^*D^4F_{\Omega}$. The matrix entries $(F_{\Omega}^*D^4F_{\Omega})_{ii}$ can be simplified to
\begin{align*}
    (F_{\Omega}^*D^4F_{\Omega})_{ii}&=\frac{1}{n}\langle D^2F_{\Omega}e_i, D^2F_{\Omega} e_i\rangle = \langle D^4 F_{\Omega} e_i,F_{\Omega} e_i\rangle = \sum\limits_{j=1}^n D_{jj}^4 \vert (F_{\Omega})_{ij}\vert^2=\sum\limits_{j=1}^n D_{jj}^4\\
    &= \operatorname{trace}(D^4) = \Vert D^2\Vert^2_F.
\end{align*}
Defining $u:=\beta^0-H^*\hat{z}$, we rewrite the remainder term as
\begin{align*}
    R=\sqrt{n}\left(\frac{(AH)^*(AH)}{n}-I_{p\times p}\right)u.
\end{align*}
For $j\in[p]$ we have
\begin{align*}
    R_j&=\sqrt{n}e_j^H\left(\frac{(AH)^*(AH)}{n}-I_{p\times p}\right)u\\
    &=\sqrt{n}\left(\frac{(AH)^*(AH)}{n}-I_{p\times p}\right)_{jj}u_j+\sqrt{n}e_j^T\left(\frac{(AH)^*(AH)}{n}-I_{p\times p}\right)_{-j}u_{-j}\\
    &=\frac{1}{\sqrt{n}}(e_j^T(AH)^*(AH)_j-1)u_j+\frac{1}{\sqrt{n}}e_j^T(AH)^*(AH)_{-j}u_{-j}\\
    &=\underbrace{\frac{1}{\sqrt{n}}((AH)^*_j(AH)_j-1)u_j}_{:=R_j^{(1)}}+\underbrace{\frac{1}{\sqrt{n}}(AH)_j^*(AH)_{-j}u_{-j}}_{R_j^{(2)}}.
\end{align*}
We obtain for $R_j^{(1)}$
$$R_j^{(1)}=\frac{1}{\sqrt{n}}\sum\limits_{i=1}^n(\overline{(AH)}_{ij}(AH)_{ij}-1)u_j:=\frac{1}{\sqrt{n}}\sum\limits_{i=1}^nZ_i^{(1)}$$
with $Z_i^{(1)}=(e_j^TH^*\overline{a}_ia_i^THe_j-1)u_j$.
Since we can write the transpose of the $i$-th row of $(AH)$ as $(a_i^TH)^T=H^Ta_i$, the mean of $Z_i^{(1)}$ is given by
\begin{align*}
    \mathbb{E}[Z_i^{(1)}]&=(e_j^T\mathbb{E}[H^*\overline{a}_ia_i^TH]e_j-1)u_j=(e_j^TH^*\mathbb{E}[a_ia_i^*]He_j-1)u_j=0.
\end{align*}
In the last step we exploited that $\mathbb{E}[a_ia_i^*]=I_{p\times p}$, because $A$ is a random sampling matrix associated to a BOS.
In order to compute the variance we start with
\begin{align*}
    \vert Z_i^{(1)}\vert^2&=Z_i^{(1)}\overline{Z_i^{(1)}}=\vert u_j\vert^2\left(e_j^TH^*\overline{a}_ia_i^THe_je_j^TH^Ta_ia_i^*\overline{H}e_j-2e_j^TH^*\overline{a}_ia_i^THe_j+1 \right).
\end{align*}
We denote the matrix $DF_{\Omega}$ by $\Tilde{F}_{\Omega}$ and hence the $i$-th row by $\Tilde{f}_{\Omega,i}$. Since the absolute value of Fourier entries are bounded by $1$, i.e. $\vert F_{\Omega,ij}\vert^2\leq 1$ we obtain
\begin{align}\label{eq:plug_in_def_a1}
    e_j^TH^*\overline{a}_ia_i^THe_j &=
    e_j^TH^*\overline{((\Tilde{f}_{\Omega,i})^TH^*)^T}(\Tilde{f}_{\Omega,i})^TH^*)He_j
    =e_j^T\overline{\Tilde{f}_{\Omega,i}}\Tilde{f}_{\Omega,i}^Te_j \\
    &=D_{ii}^2e_j^T\overline{f_{\Omega,i}}f_{\Omega,i}^Te_j 
    \leq  D_{ii}^2 \vert F_{\Omega,ij}\vert^2 
    \leq D_{ii}^2.\label{eq:plug_in_def_a2}
\end{align}
Hence
$$e_j^TH^*\overline{a}_ia_i^THe_j e_j^TH^Ta_ia_i^*\overline{H}e_j\leq D_{ii}^4.$$
The expectation reads as
\begin{align*}
    \mathbb{E}[D_{ii}^4]=\sum\limits_{i=1}^p\nu(i)D_{ii}^4=\sum\limits_{i=1}^p\frac{\kappa_i^2}{\Vert \kappa\Vert_2^2}\left(\frac{\Vert \kappa\Vert_2}{\sqrt{p}\kappa_i}\right)^4=\frac{\Vert \kappa\Vert_2^2}{p^2}\sum\limits_{i=1}^p\frac{1}{\kappa_i^2}
\end{align*}
and with $\kappa_i=\frac{3\sqrt{2\pi}}{\sqrt{i}}$ \cite[Corollary 2]{6651836} it becomes
$$\mathbb{E}[D_{ii}^4]=\frac{\Vert \kappa\Vert_2^2}{18\pi p^2}\sum\limits_{i=1}^p i=\frac{\Vert \kappa\Vert_2^2}{18\pi p^2}\frac{p^2+p}{2}\leq\frac{\Vert\kappa\Vert_2^2}{36\pi}+1.$$
In summary, we obtain
\begin{align*}
    \mathbb{E}[\vert Z_i^{(1)}\vert^2]&=\vert u_j\vert^2\left(-1+\mathbb{E}[e_j^TH^*\overline{a}_ia_i^THe_je_j^TH^Ta_ia_i^*\overline{H}e_j]\right)\\
    &\leq \vert u_j\vert^2\left(-1+\mathbb{E}[D_{ii}^4]\right)\leq \frac{\Vert\kappa\Vert_2^2}{36\pi} \vert u_j\vert^2.
\end{align*}
Using the same calculation as in \eqref{eq:plug_in_def_a1} and \eqref{eq:plug_in_def_a2}, and additional, that $\kappa_i\geq \frac{3\sqrt{2\pi}}{\sqrt{p}}$, we bound 
\begin{align*}
\vert Z_i^{(1)}\vert &\leq \vert u_j\vert\cdot \vert D_{ii}^2-1\vert\leq \vert (H^*(z^0-\hat{z}))_j\vert \cdot \left\vert\frac{\Vert \kappa\Vert_2^2}{p \kappa_i^2}-1\right\vert\leq \Vert H_j^*\Vert_{\infty}\Vert z^0-\hat{z}\Vert_1\cdot \left\vert\frac{\Vert \kappa\Vert_2^2}{18\pi}-1\right\vert\\
&\leq \Vert z^0-\hat{z}\Vert_1\cdot \left\vert\frac{\Vert \kappa\Vert_2^2}{18\pi}-1\right\vert.
\end{align*}
For the estimate of $R_j^{(2)}$ we introduce the shorthand notation $u^{(-j)}$ to denote the vector $u$ with its $j$-th component set to zero and obtain
\begin{align*}
    R_j^{(2)}=\frac{1}{\sqrt{n}}\sum\limits_{i=1}^n\overline{(AH)}_{ij}\langle e_i^T(\overline{(AH)},u^{(-j)}\rangle=\frac{1}{\sqrt{n}}\sum\limits_{i=1}^n Z_i^{(2)},
\end{align*}
with $Z_i^{(2)}=e_j^TH^*\overline{a}_i\langle H^*\overline{a}_i,u^{(-j)}\rangle$.
The mean of $Z_i^{(2)}$ is
$$\mathbb{E}[Z_i^{(2)}]=e_j^TH^*\mathbb{E}[\overline{a}_i\overline{a}_i^*]Hu^{(-j)}= \langle e_j, u^{(-j)} \rangle=0.$$
We calculate
\begin{align*}
    \vert Z_i^{(2)}\vert^2&=Z_i^{(2)}\overline{Z_i^{(2)}}
    =e_j^TH^*\overline{a}_i\overline{a}_i^*Hu^{(-j)}e_j^TH^Ta_i\langle u^{(-j)}, H^*\overline{a}_i\rangle\\
    &=\langle u^{(-j)}, H^*\overline{a}_i\rangle e_j^TH^*\overline{a}_ie_j^TH^Ta_i\overline{a}_i^*Hu^{(-j)}\\
    &=(u^{(-j)})^*H^*\overline{a}_ie_j^TH^*\overline{((\Tilde{f}_{\Omega,i})^TH^*)^T}e_j^TH^T((\Tilde{f}_{\Omega,i})^TH^*)^T\overline{a}_i^*Hu^{(-j)}\\
    &=(u^{(-j)})^*H^*\overline{a}_ie_j^T\overline{\Tilde{f}_{\Omega,i}}e_j^T\Tilde{f}_{\Omega,i}\overline{a}_i^*Hu^{(-j)}
    =(u^{(-j)})^*H^*\overline{a}_iD_{ii}^2e_j^T\overline{f_{\Omega,i}}e_j^T f_{\Omega,i}\overline{a}_i^*H u^{(-j)}\\
    &=(u^{(-j)})^*H^*\overline{a}_iD_{ii}^2\vert F_{\Omega,ij}\vert^2\overline{a}_i^*Hu^{(-j)}
    \leq \max\limits_{i\in[p]}D_{ii}^2(u^{(-j)})^*H^*\overline{a}_i\overline{a}_i^*Hu^{(-j)}.
\end{align*}
The maximum of $D_{ii}^2$ is attained for $i=p$ with value $\max\limits_{i\in[p]}D_{ii}^2=\max\limits_{i\in[p]}\left(\frac{\Vert \kappa\Vert_2}{\sqrt{p}\kappa_i}\right)^2=\frac{\Vert\kappa\Vert_2^2}{18\pi}$.
Thus, we obtain
$$\mathbb{E}[\vert Z_i^{(2)}\vert^2]\leq \max\limits_{i\in[p]}D_{ii}^2(u^{(-j)})^*H^*\mathbb{E}[\overline{a}_i\overline{a}_i^*]Hu^{(-j)}\leq \frac{\Vert\kappa\Vert_2^2}{18\pi}\cdot \Vert u\Vert_2^2.$$
Finally,

\begin{align*}
    \vert Z_i^{(2)}\vert&= \vert e_j^TH^*\overline{a}_i\overline{a}_i^*Hu^{(-j)}\vert
    \leq D_{ii}\vert e_j^T\overline{f}_{\Omega,i}\vert\cdot\vert \overline{a}_i^*Hu^{(-j)}\vert \\
    &=\frac{\Vert\kappa\Vert_2}{\sqrt{18\pi}}\vert F_{\Omega,ij}\vert\cdot\vert \overline{a}_i^*H(H^*z^0-H^*\hat{z})\vert
    \leq \frac{\Vert\kappa\Vert_2}{\sqrt{18\pi}}\vert \overline{a}_i^*(z^0-\hat{z})\vert \\
    &\leq \frac{\Vert\kappa\Vert_2}{\sqrt{18\pi}}\vert\cdot \Vert \overline{a}_i^*\Vert_{\infty}\Vert z^0-\hat{z}\Vert_1
    \leq \frac{\Vert\kappa\Vert_2^2}{\sqrt{18\pi}}\cdot \Vert z^0-\hat{z}\Vert_1.
\end{align*}
In summary we have
\begin{align*}
    &\mathbb{E}[Z_i] = \mathbb{E}[Z_i^{(1)}] + \mathbb{E}[Z_i^{(2)}] = 0,\\
    &\mathbb{E}[\vert Z_i\vert^2] \leq 2 \left(\mathbb{E}[\vert Z_i^{(1)}\vert^2] + \mathbb{E}[\vert Z_i^{(2)}\vert^2]\right) \leq 2\left(\frac{\Vert \kappa \Vert_2^2}{36\pi}\vert u_j\vert^2 + \frac{\Vert \kappa \Vert_2^2}{18\pi} \cdot \Vert u \Vert_2^2\right)\leq \frac{2}{9\pi}\Vert \kappa \Vert_2^2\cdot \Vert u \Vert_2^2, \\
    &\vert Z_i\vert \leq \vert Z_i^{(1)}\vert + \vert Z_i^{(2)}\vert \leq \Vert z^0-\hat{z}\Vert_1\left(\left\vert \frac{\Vert \kappa\Vert_2^2}{18\pi}-1 \right\vert + \frac{\Vert \kappa\Vert_2^2}{\sqrt{18\pi}}\right) \leq \frac{2}{\sqrt{18\pi}}\Vert \kappa\Vert_2^2\cdot \Vert z^0-\hat{z}\Vert_1.
\end{align*}
The orthogonality of $H$ yields
$$\Vert u\Vert_2=\Vert \beta^0-H^*\hat{z}\Vert_2=\Vert H^*z^0-H^*\hat{z}\Vert_2=\Vert z^0-\hat{z}\Vert_2.$$
In contrast to Theorem \ref{thm:mainresult} we now have the model $Dy = Az^0 + D\varepsilon$. Therefore, we need to condition on the event 
$$\mathcal{F}:=\{\varepsilon\in\mathbb{C}^n: \max\limits_{j\in [p]}\frac{2}{n}\vert \langle D\varepsilon, A_j\rangle \vert\} = \{\varepsilon\in\mathbb{C}^n: \max\limits_{j\in [p]}\frac{2}{n}\vert \langle \varepsilon, DA_j\rangle \vert\},$$
where we exploited, that $D$ is self-adjoint. In order to find a proper choice for $\lambda_0$ we need to replace Equation \eqref{eq:bound_column_design} by
$$\Vert DA_j\Vert_2\leq \max\limits_{l\in[n]}\frac{\Vert \kappa \Vert_2}{\sqrt{p}\kappa_l}\cdot \Vert A_j\Vert_2\leq \frac{\Vert \kappa \Vert_2}{\sqrt{18\pi}}\cdot \Vert A_j\Vert_2,$$
where we used, that $\kappa_i\geq \frac{3\sqrt{2\pi}}{\sqrt{p}}$. Thus, the choice of $\lambda_0$ is in this case $\frac{\sigma}{3\sqrt{\pi}}\frac{\Vert \kappa \Vert_2^{3/2}}{\sqrt{n}}(\sqrt{2}+\sqrt{10 \log p})$.
Therefore the tail bound for the remainder term $R$ is derived by Bernstein's inequality in the same way as in Theorem \ref{thm:mainresult} using Lemma \ref{6.2}, Lemma \ref{le:l2oracleinequality}, Theorem \ref{thm:LASSOforRIP} and Theorem \ref{thm:sparse_transform_RIP}:
\begin{equation*}
    \mathbb{P}\left(\Vert R\Vert_{\infty}\geq C(\Vert\kappa\Vert_2,\sigma,\delta_t)\frac{\sqrt{s_0}\log p}{\sqrt{n}}\right)\leq 5p^{-2}+p^{-c\log^3 t},
\end{equation*}
where $C(\Vert\kappa\Vert_2,\sigma,\delta_t)>0$ is a constant depending on $\Vert\kappa\Vert_2$, $\sigma$ and the RIP constant $\delta_t$.
\end{proof}

For the sake of completeness, we provide a formula for the radius of the confidence circles $J^{\circ}_i(\alpha):=\{z\in\mathbb{C}:\vert\hat{\beta}_i^u-z\vert\leq\delta^{\circ}(\alpha)\}$ for a non-sparse ground truth $\beta^0$, estimated via the desparsified LASSO \eqref{eq:desparsified_LASSO_sparsetrafo} with significance level $\alpha\in(0,1)$, in the setting of Theorem \ref{thm:mainresult_sparse_transform}. The construction is straightforward to the one in Theorem \ref{thm:confidence circle} and the radius is given by
$$\delta^{\circ}(\alpha):=\frac{\hat{\sigma}\Vert D^2\Vert_F}{n}\sqrt{\log\left(\frac{1}{\alpha}\right)},$$
where $\hat{\sigma}$ is a noise estimator satisfying \eqref{eq:consistent_noise}.

\section{Numerics and application to MRI}\label{sec:numerics}

In this section, we illustrate our theoretical findings with numerical experiments. We use TFOCS \cite{becker2011templates} for the c-LASSO estimator. We conduct our experiments by choosing a subsampled Fourier matrix as a design matrix. In this setting, we exploit the FFT for calculating the desparsified LASSO in a fast and scalable way. Throughout our experiments, we assume, for simplicity, that the noise level $\sigma$ is known. The numerical results are still consistent if the methods for estimating the noise level from data that were discussed at the end of Section \ref{sec:confidence regions} are employed. Furthermore, we set the confidence level $\alpha=0.05$.

\subsection{Synthetic data}\label{subsec:synthetic_data}
We consider the linear model \eqref{Cmodel} with dimensions $p=100000,\, n=0.4p$ and $s_0=0.01 p$. Although the theory requires a subsampled Fourier matrix with rows that are independently and uniformly selected from $[p]$, in practice we sample without replacement. This sampling pattern leads to only slightly different recovery results as discussed in \cite[Section 12.6]{Foucart.2013}.  Also for the $s_0$-sparse ground truth $\beta^0\in\mathbb{C}^{p}$ we select a subset $S$ of $s_0$ different indices uniformly at random from $[p]$. The entries of $\beta^0$ in the support $S_0$ are complex standard Gaussian distributed and the vector is normalized, i.e., $\Vert\beta^0\Vert_2=1$. The noise is complex Gaussian distributed with noise level $\sigma=0.15$, which corresponds to a relation $\frac{\Vert\varepsilon\Vert}{\Vert F_{\Omega}\beta^0\Vert}\approx 0.15$. The regularization parameter for c-LASSO is set as $\lambda=25\lambda_0$. Note, that $\lambda_0$ depends on the dimensions. This specific choice differs from the sufficient condition $\lambda\geq 2\lambda_0$ and is explained in Section \ref{subsec:choice of lambda}.

The first experiment aims to confirm the Gaussian distribution of $\sqrt{n}(\hat{\beta}^u-\beta^0)$ numerically, as stated in the main result Theorem \ref{thm:mainresult}. In a Q-Q plot we compare the quantiles of $\frac{\sqrt{2n}}{\hat{\sigma}}\Re(\hat{\beta}^u-\beta^0)$ versus the quantiles of a standard normal distribution. Note that in the partial Fourier case $\hat{\Sigma}_{ii}=1$ for all $i\in[p]$. Figure \ref{fig:qqplot} shows that the quantiles basically fall onto the identity line, which is a strong hint, that both have the same distribution. The same behavior holds for the imaginary part.

\begin{figure}
    \centering
    \includegraphics[width=0.5\textwidth]{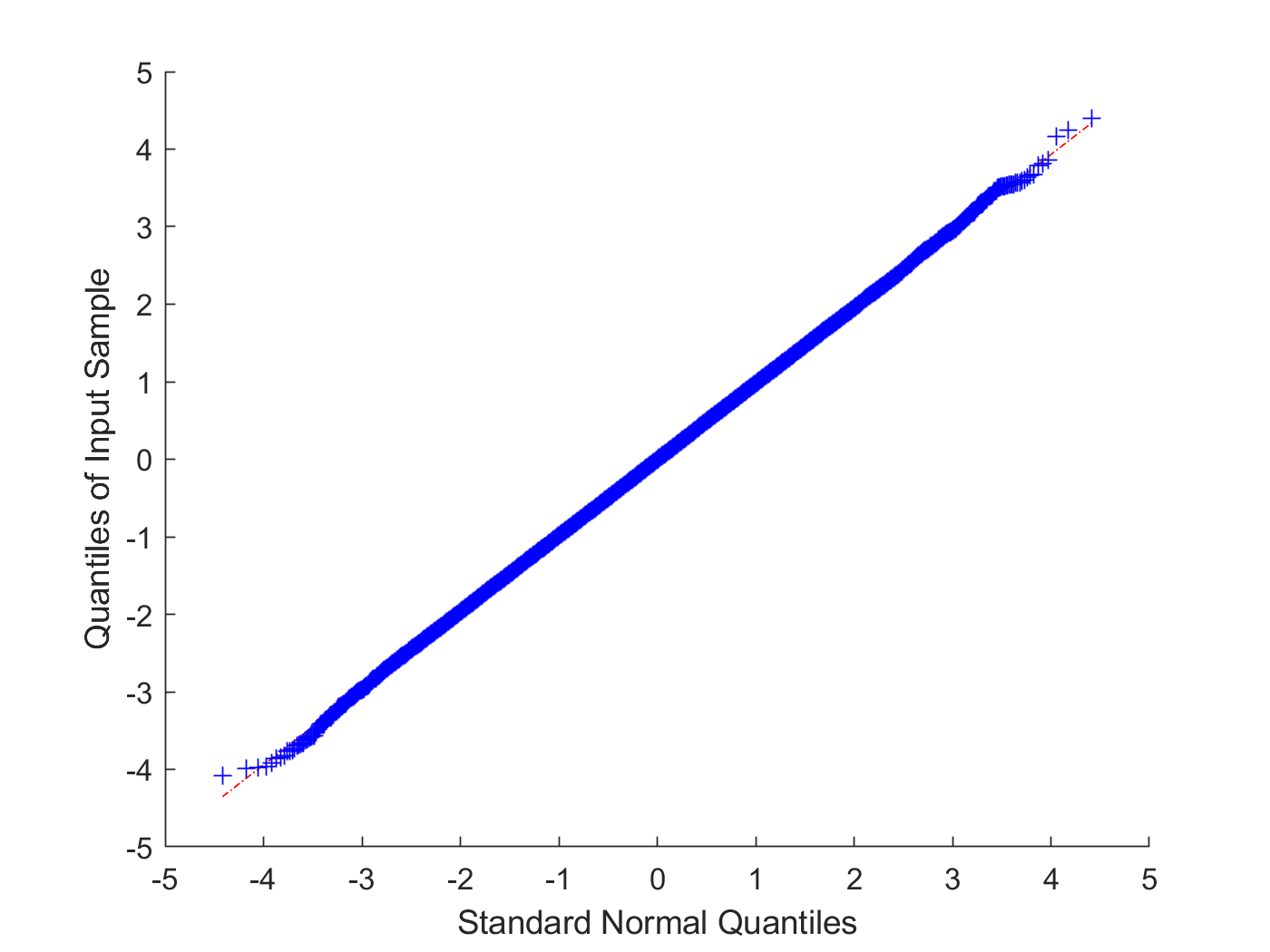}
    \caption{Q-Q plot for one realization of $\frac{\sqrt{2n}}{\hat{\sigma}}\Re(\hat{\beta}^u-\beta^0)$ as an input sample.}
    \label{fig:qqplot}
\end{figure}

The second experiment is conducted to confirm Theorem \ref{thm:confidence circle}. We use the same setting as above with $100$ independent realizations of the noise. We estimate the probability $\mathbb{P}(\beta^0_i\in J_i^{\circ}(\alpha))$ with
$$\widehat{\mathbb{P}}_i(\alpha)=\frac{1}{100}\sum\limits_{i=1}^{100}\mathbbm{1}_{\{\beta^0_i\in J_i^{\circ}(\alpha)\}}.$$

In the next step, we average this value while distinguishing between $i\in S_0$ and $i\in [p]$. We define the hitrate and the hitrate on the support, respectively, as $$h(\alpha)=\frac{1}{p}\sum\limits_{i=1}^p\mathbb{P}(\beta^0_i\in J_i^{\circ}(\alpha)),\qquad h_{S_0}(\alpha)=\frac{1}{s_0}\sum\limits_{i\in S_0}\mathbb{P}(\beta^0_i\in J_i^{\circ}(\alpha)).$$

We achieve the hitrates $h(0.05)=0.9554$ and $h_{S_0}(0.05)=0.9366$, which are very close to our theoretical prediction of a $95\%$ confidence.

\subsection{Non-sparse synthetic data}

In order to confirm Theorem \ref{thm:mainresult_sparse_transform} we set up the following experiment. We set $p=2^{15}=32768$ with the model $y=F_{\Omega}\beta^0+\varepsilon$, where $F_{\Omega}\in\mathbb{C}^{n \times 32768}$ is a subsampled Fourier matrix ($n$ is explained below), where the rows are sampled according to the probability measure from Theorem \ref{thm:mainresult_sparse_transform}, with $\kappa_j = \min\left(\frac{3\sqrt{2\pi}}{\sqrt{j}},1\right)$. We choose $\beta^0$ such that $z^0=H\beta^0$, where $H$ is the Haar transform and $z^0_i=1$ for $i\in[100]$. The noise is distributed according to a complex Gaussian with relative noise level in the Haar domain $\frac{\Vert \varepsilon\Vert_2}{\Vert F_{\Omega}H^*z^0\Vert_2}\approx 3.3\%$.

The $n=0.95p$ rows of $F_{\Omega}$ are sampled with respect to $\nu$ with replacement, which leads to around $n_{eff}=0.34p$ pairwise distinct rows. The remaining rows are sampled more than once. Although the sampling is done with replacement we use reweighted sampling without replacement in order to solve the LASSO. The sampling scheme is described in detail in \cite{sampling2023}.
The regularization parameter for the LASSO is set as $\lambda=85\lambda_0$, $\lambda_0=\frac{\sigma \Vert \kappa\Vert^{3/2}}{3\sqrt{\pi}\sqrt{n}}(2+\sqrt{10\log p})$, cf. Lemma \ref{6.2}.
This experiment aims to numerically confirm the Gaussianity of $\sqrt{n}(\hat{\beta}^u-\beta^0)$, as (asymptotically) proven in Theorem \ref{thm:mainresult_sparse_transform}. As above, we compare in a Q-Q plot the quantiles of $\frac{\sqrt{2n}}{\hat{\sigma}\hat{\Sigma}_{11}}\Re(\hat{\beta}^u-\beta^0)$ versus the quantiles of a standard normal distribution, where $\hat{\Sigma}_{ii} = \frac{1}{n}(F_{\Omega}^*D^4F_{\Omega})_{ii}$. Note that $\hat{\Sigma}_{ii}=\hat{\Sigma}_{11}=\frac{1}{n}\Vert D^2\Vert_F^2$ for all $i\in[p]$. Figure \ref{fig:qqplot_sparse_trafo} confirms that the desparsified LASSO has an (asymptotic) normal distribution.

\begin{figure}
    \centering
    \includegraphics[width=1\textwidth]{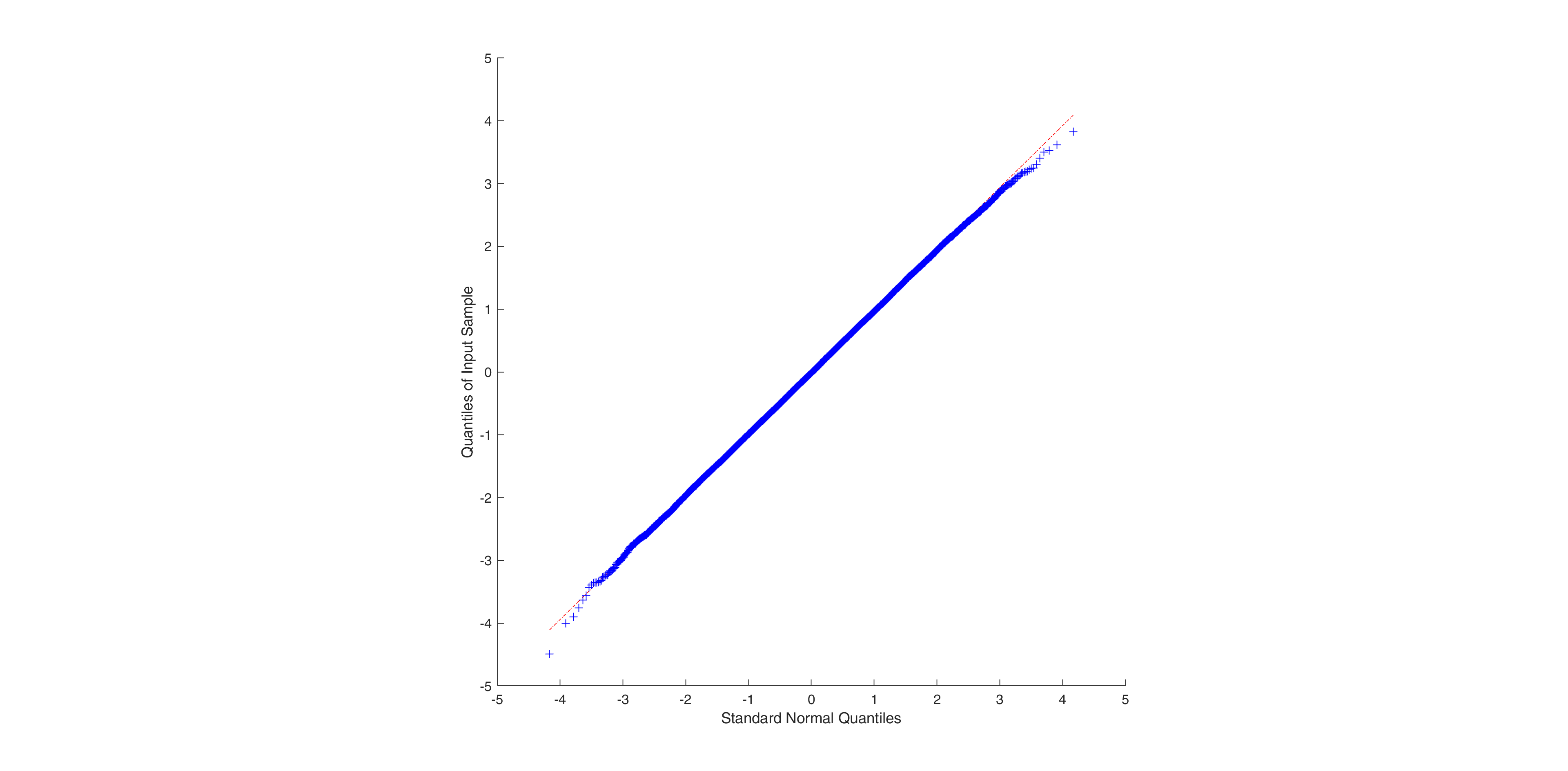}
    \caption{Q-Q plot for one realization of $\frac{\sqrt{2n}}{\hat{\sigma}\hat{\Sigma}_{11}}\Re(\hat{\beta}^u-\beta^0)$ as an input sample.}
    \label{fig:qqplot_sparse_trafo}
\end{figure}

\subsection{Choice of \texorpdfstring{$\lambda$}{lambda} via cross-validation}

One of the numerical issues in the LASSO solution is the choice of the regularization parameter $\lambda$. As it is common in the literature, throughout this paper we assume that $\lambda<\Vert X^Ty\Vert_{\infty}$. Otherwise, the KKT conditions of the LASSO lead only to the trivial solution $\hat{\beta}=0$ \cite{giraud2021introduction}. Arguably, the most common practical method for tuning the LASSO parameter is cross-validation. 

Recently, \cite{Chetverikov.07.05.2016} provided strong guarantees for the LASSO solution when $\lambda$ is chosen via cross-validation. In this procedure the data $(X,y)\in\mathbb{C}^{n\times p}\times\mathbb{C}^n$ is divided into $K$ parts. From these $K$ parts, $K-1$ parts serve as training data and one part serves as test data for $\lambda\in\Lambda_n$, where $\Lambda_n$ is the set of test parameters $\lambda$. Then, the roles of training and test data are exchanged, such that every part of the data will be part of the test dataset at some stage. The prediction errors for the test sets are added. The procedure is repeated for different values of $\lambda$. Finally, the $\lambda$ corresponding to the smallest prediction error is chosen. More formally, let $(I_k)_{k=1}^K$ be a partition of the index set $[n]$. The LASSO observed with $K-1$ data parts is defined as
\begin{equation*}
    \hat{\beta}_{-k}(\lambda)=\argmin\limits_{b\in\mathbb{C}^{p}}\frac{1}{n-\vert I_k\vert}\sum\limits_{i\notin I_k}(y_i-x_i^*b)^2+\lambda\Vert b\Vert_1.
\end{equation*}
With this, the regularization parameter is chosen as
\begin{equation*}\label{eq:cross validation}
    \hat{\lambda}=\argmin\limits_{\lambda\in\Lambda_n}\sum\limits_{k=1}^K\sum\limits_{i\in I_k}\left(y_i-x_i^*\hat{\beta}_{-k}(\lambda)\right)^2.
\end{equation*}
This regularization parameter is optimal in the sense, that it leads to a near optimal prediction norm bound \cite[Thm. 4.1]{Chetverikov.07.05.2016}. 

\subsection{Choice of regularization parameter}\label{subsec:choice of lambda}
The choice of the regularization parameter $\lambda$ has a huge impact on the quality of the confidence intervals. In order to measure this quality precisely we define the hitrate and the hitrate on the support $S_0$, respectively, as 
\begin{equation}
    h=\frac{1}{p}\sum\limits_{i=1}^p\mathbbm{1}_{\{\beta^0_i\in J_i^{\circ}\}},\quad h_{S_0}=\frac{1}{s_0}\sum\limits_{i\in S}\mathbbm{1}_{\{\beta^0_i\in J_i^{\circ}\}}.
\end{equation}
The theoretical analysis in Theorem \ref{thm:LASSOforRIP} of the LASSO requires $\lambda\geq 2\lambda_0$. Thus, $\lambda=2\lambda_0$ would be a sufficient choice for having a small bound in the difference between the LASSO solution and the ground truth, but it is a priori not clear whether this choice is optimal for UQ, in the sense that it is not clear if we would obtain the highest hitrates on the support $h_{S_0}$. Therefore, we examine the hitrates of different values of $\lambda$. 

We calculate for every $\lambda = \frac{k\cdot\lambda_0}{4}, \,k\in[80]$, the hitrates $h(0.05)$ and $h_{S_0}(0.05)$. We keep the setting described in Section \ref{subsec:synthetic_data} but instead of $p=100000$ we set $p=1000$. We conduct the experiment for a partial Fourier matrix and for a standard complex Gaussian matrix. In Figure \ref{fig:gaussian_lambda_hitrates} and \ref{fig:fourier_lambda_hitrates} we plot the hitrates along $\lambda$. It is not a surprise that the hitrates in the Gaussian case are slightly better. Interestingly, in both cases, the choice $\lambda\approx 2\lambda_0$ leads to high hitrates for $h$ and $h_{S_0}$. But the hitrate on the support $h_{S_0}$ increases to a maximum around $\lambda=3\lambda_0$ before it decreases for large $\lambda$. In both cases, in the Gaussian and partial Fourier case the hitrates $h$ and $h_{S_0}$ seem to coincide around $\lambda=3\lambda_0$. In addition to the hitrates we plot the error of cross validation based on \eqref{eq:cross validation} with $K=5$. The optimal $\lambda$ based on cross validation closely coincides with the regularization parameter that provides the highest rate $h_{S_0}$.

\begin{figure}
    \centering
    \begin{minipage}{0.45\textwidth}
        \centering
        \includegraphics[width=1\textwidth]{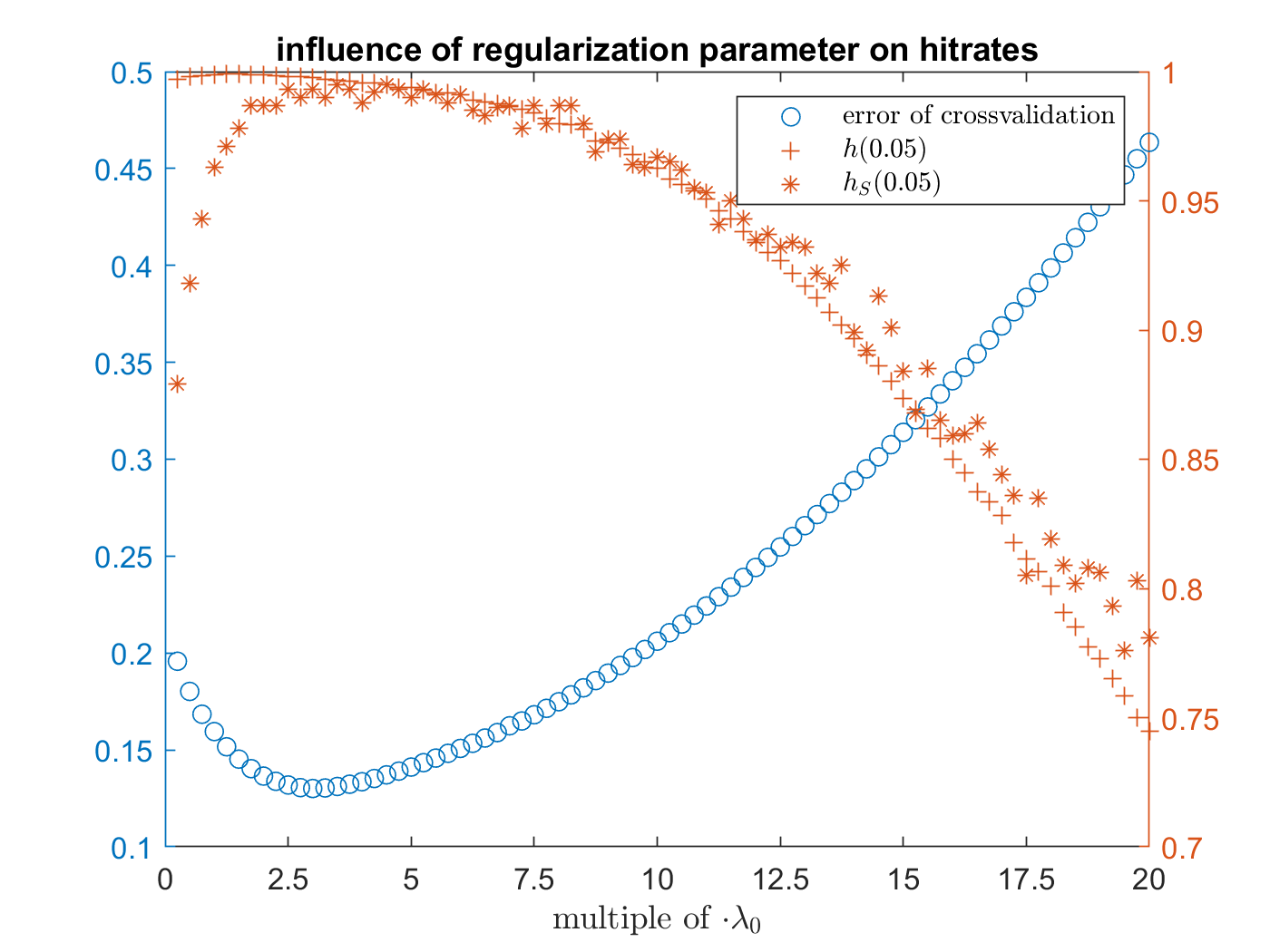} 
        \caption{Hitrates and error of cross validation in relation to the regularization parameter for Gaussian design with $p=1000$.}\label{fig:gaussian_lambda_hitrates}
    \end{minipage}\hfill
    \begin{minipage}{0.45\textwidth}
        \centering
        \includegraphics[width=1\textwidth]{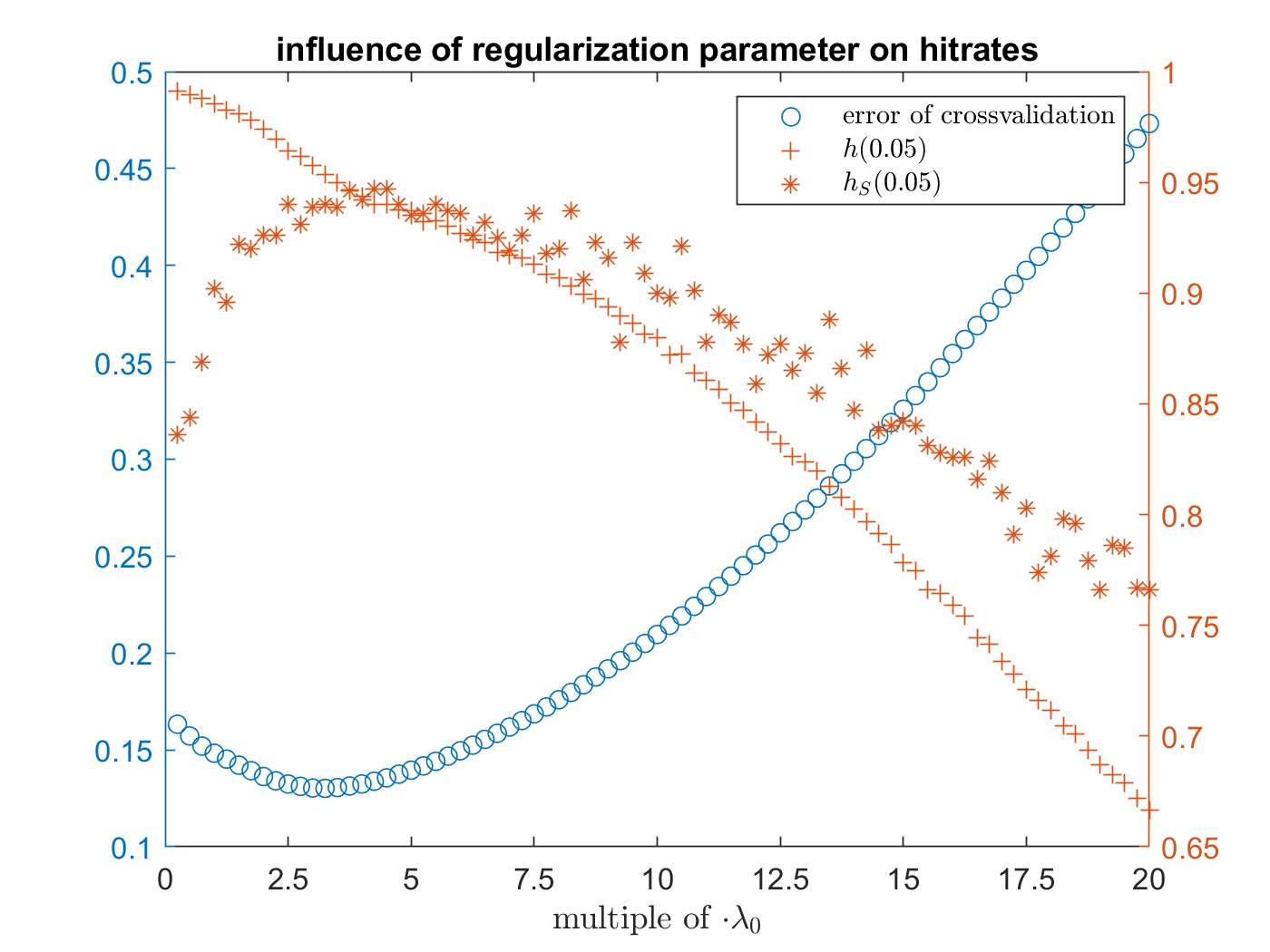}
        \caption{Hitrates and error of cross validation in relation to the regularization parameter for subsampled Fourier design with $p=1000$.}\label{fig:fourier_lambda_hitrates}
    \end{minipage}
\end{figure}

The same observation can be made for higher dimensions, i.e., in the regime $p=10000$ and $p=100000$, respectively. Figure \ref{fig:10k_fourier_lambda_hitrates} and \ref{fig:100k_fourier_lambda_hitrates} show, that the optimal regularization parameter for which the highest hitrate $h_{S_0}$ is achieved is around $\lambda=10\lambda_0$ and $\lambda=25\lambda_0$, respectively. Currently, we do not have an explanation, why these factors are larger for larger $p$. For a smaller $\lambda$ we have again a similar trade-off as before. We achieve a high hitrate over all entries, but a low hitrate for the entries within the support. The rates seem to coincide above the threshold $\lambda=28\lambda_0,$ which is in the region of the optimal regularization parameter regarding cross-validation.

\begin{figure}
    \centering
    \begin{minipage}{0.45\textwidth}
        \centering
        \includegraphics[width=1\textwidth]{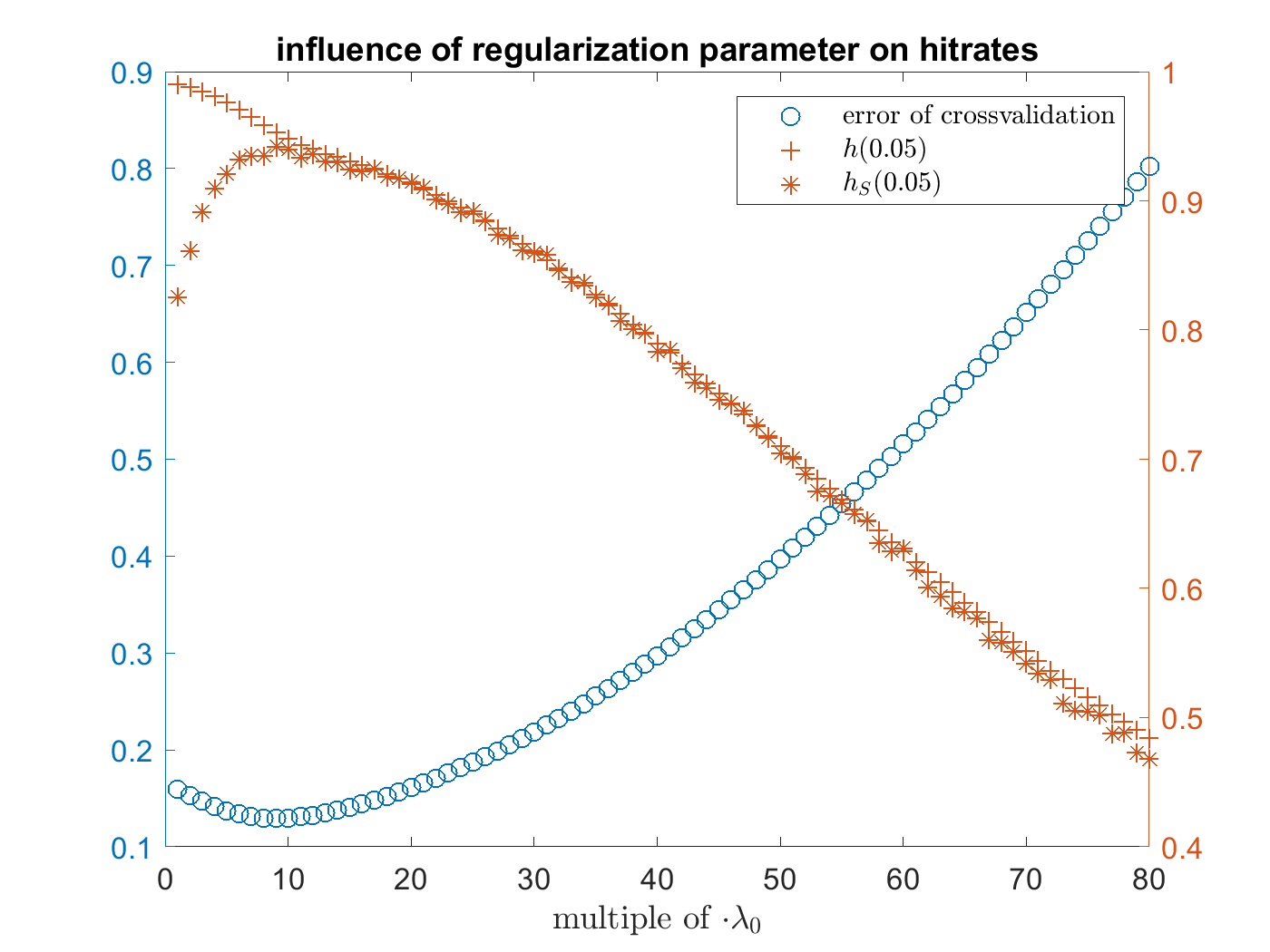}
        \caption{Hitrates and error of cross validation in relation to the regularization parameter for subsampled Fourier design with $p=10000$.}\label{fig:10k_fourier_lambda_hitrates}
    \end{minipage}\hfill
    \begin{minipage}{0.45\textwidth}
        \centering
        \includegraphics[width=1\textwidth]{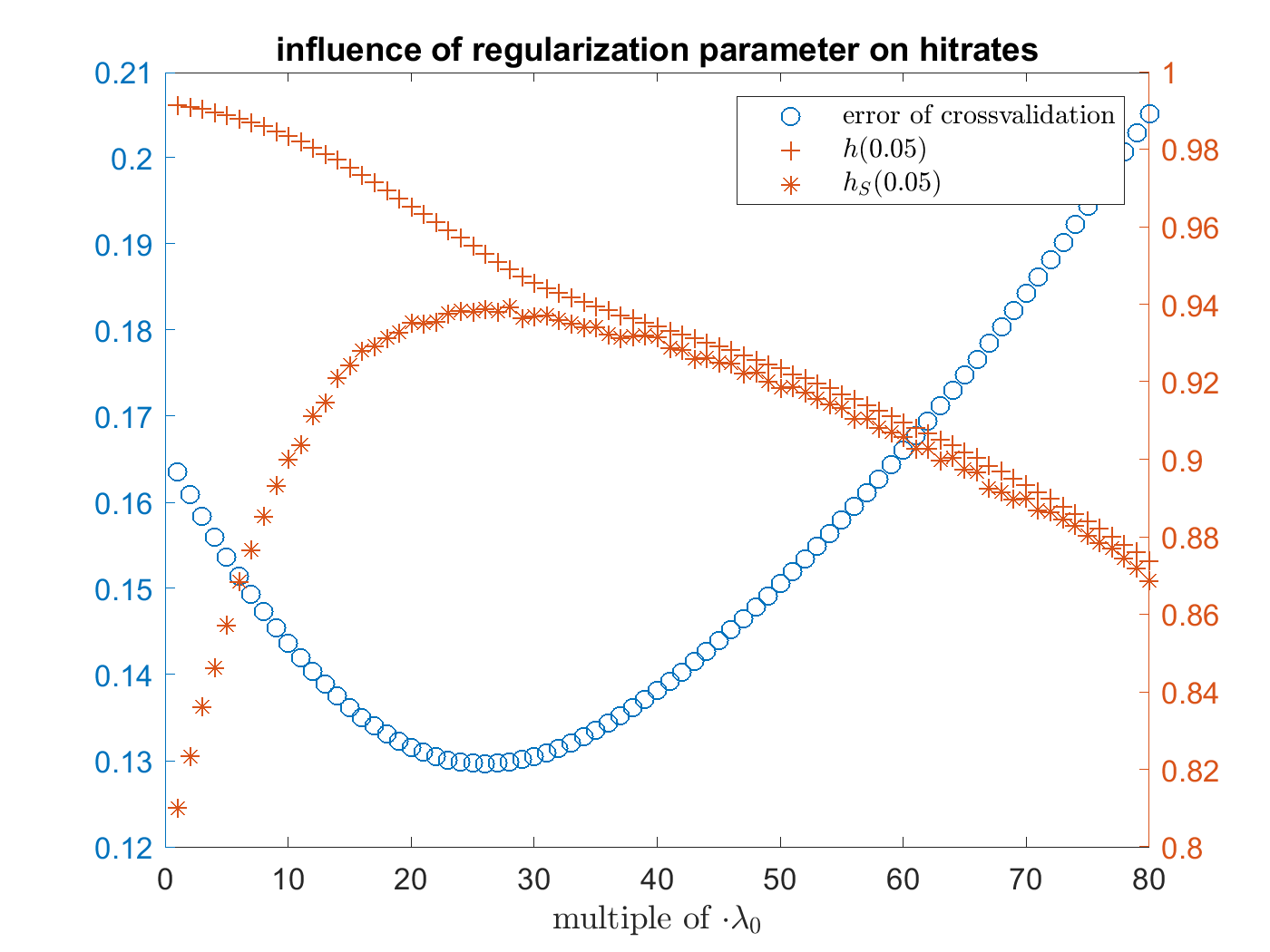}
        \caption{Hitrates and error of cross validation in relation to the regularization parameter for subsampled Fourier design with $p=100000$.}\label{fig:100k_fourier_lambda_hitrates}
    \end{minipage}
\end{figure}

\subsection{Real data}
In order to conclude our experiments, we use real-world data for the validation of our contribution. To this end, we first apply our concept to naturally sparse MR angiography brain imaging data, which only have a few image pixels with significant, non-zero entries, corresponding to the brain vessels that light up \cite{lustig}. Figure \ref{fig:non_sparse_angio} depicts an exemplary MR angiography image of these brain vessels taken from the Brain Vasculature (BraVa) database \cite{Wright.2013}, which serves as input for our simulated data to test the performance and limits of our contribution. Experimentally, we artificially increase the sparsity of the input image - since we are only interested in the vessels - by setting a manual threshold removing the brain background image information and noise floor. This means that every pixel with a magnitude lower than this threshold is set to $0$ as presented in Figure \ref{fig:sparse_angio}. And it corresponds to an image intensity threshold of 200 [a.u.] of the input angiogram that results in a $s_0=1282$-sparse image, which serves as the (unknown) ground truth $\beta^0\in\mathbb{R}^{92160}$. Note that the theory as well as the algorithm allow complex images, but due to better illustration, we illustrate the procedure with real images where we simply set the imaginary part equal to $0$.

\begin{figure}
     \centering
     \begin{subfigure}[b]{0.31\textwidth}
         \centering
         \includegraphics[width=\textwidth]{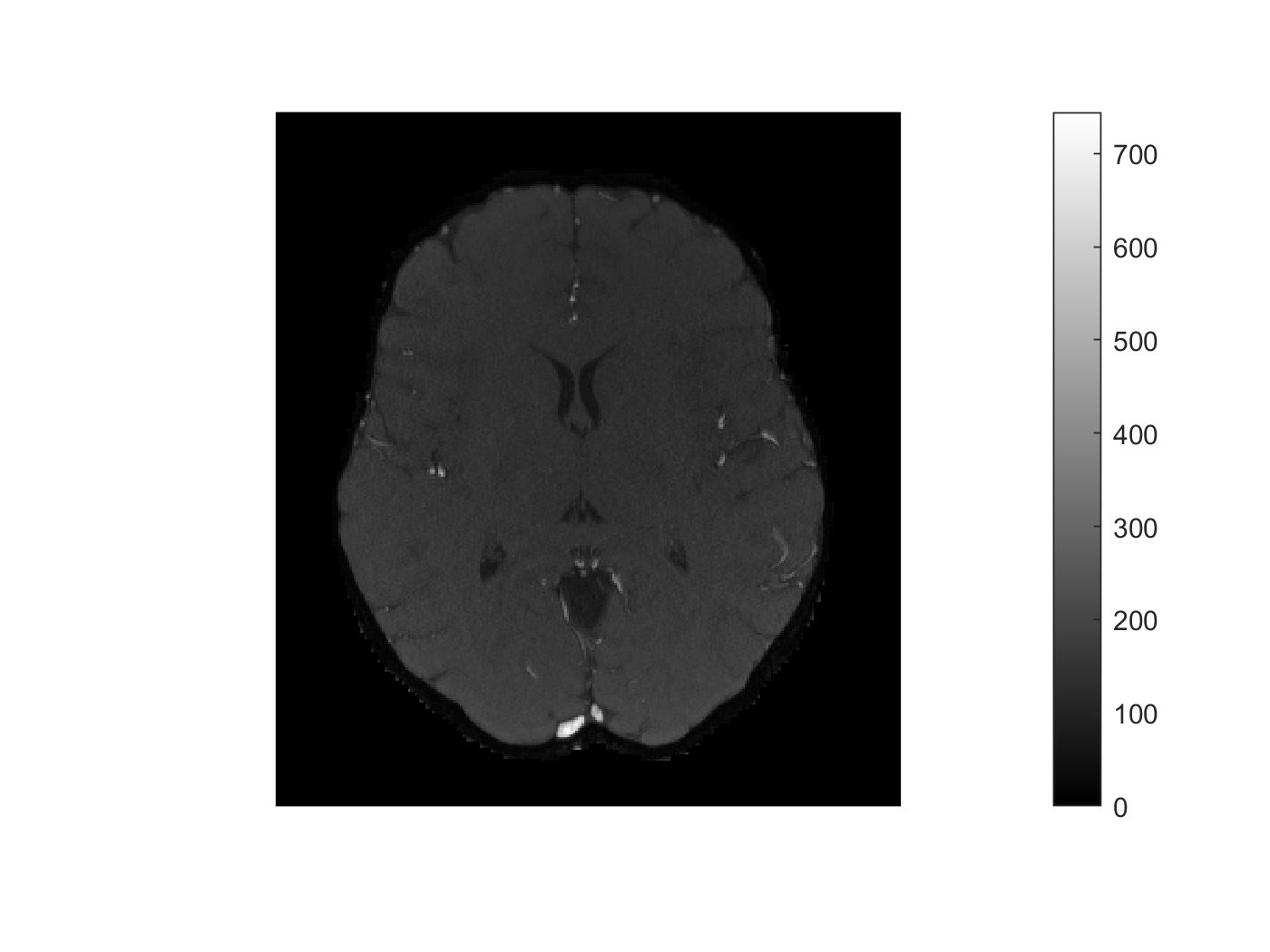}
         \caption{Original non-sparse angiography data of single brain slice with vessels lighting up; image intensities in arbitrary units [a.u.].}
         \label{fig:non_sparse_angio}
     \end{subfigure}
     \hfill
     \begin{subfigure}[b]{0.31\textwidth}
         \centering
         \includegraphics[width=\textwidth]{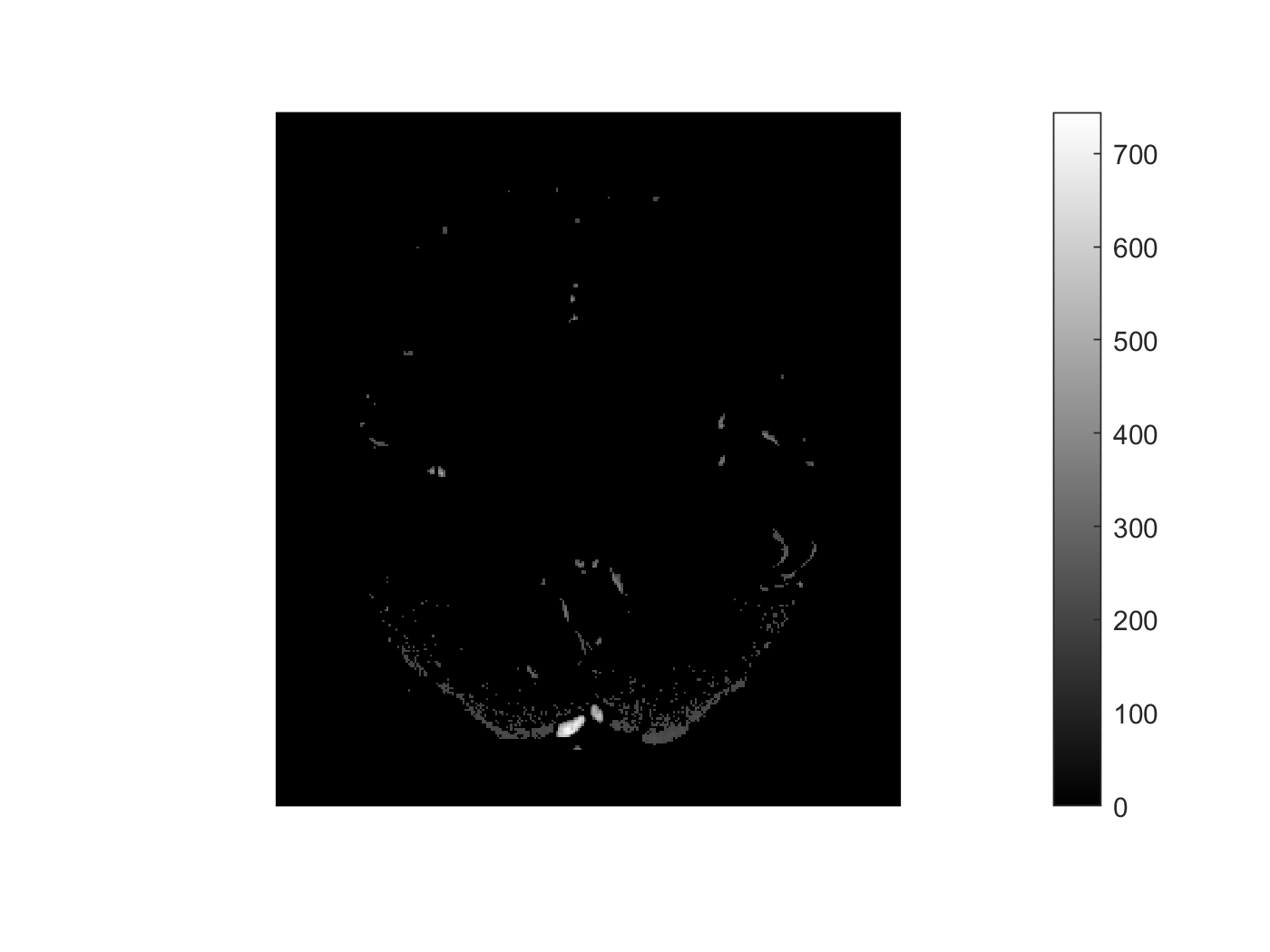}
         \caption{Sparsified angiography ($s_0=1282$) data after thresholding with $200$ [a.u.]; image intensities in arbitrary units [a.u.].}
         \label{fig:sparse_angio}
     \end{subfigure}
     \hfill
     \begin{subfigure}[b]{0.31\textwidth}
         \centering
         \includegraphics[width=\textwidth]{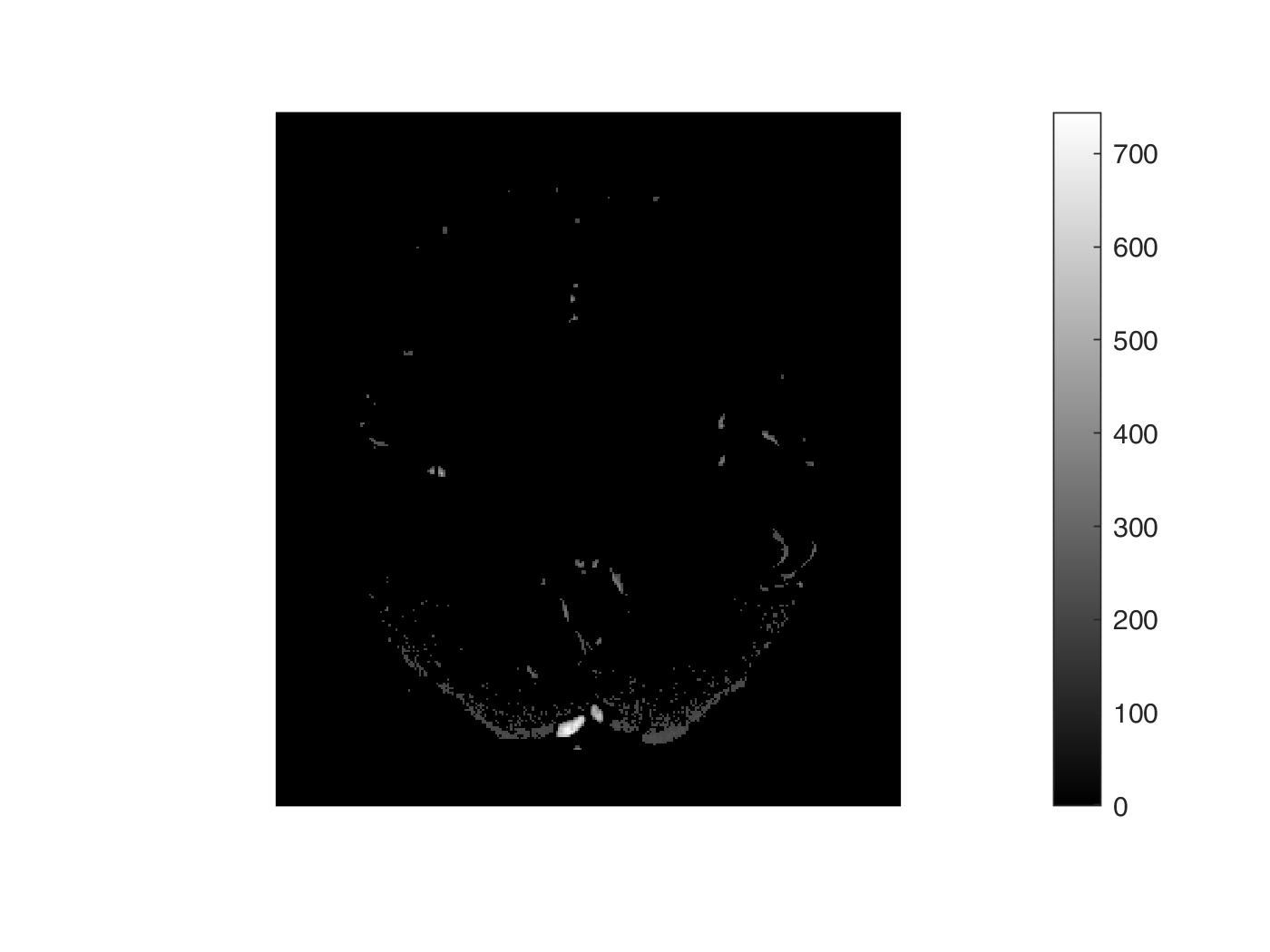}
         \caption{Recovered image via desparsified LASSO of sparse angiography ($s_0=1282$); image intensities in arbitrary units [a.u.].}
         \label{fig:recov_sparse_angio}
     \end{subfigure}
        \caption{Angiography data}
        \label{fig:three graphs}
\end{figure}

In the following experiment, we simulate the image acquisition in MRI by subsampling $n=0.4p=36864$ different rows of a full Fourier matrix. Then, we add complex noise with $\sigma=1000$ leading to $\frac{\Vert\varepsilon\Vert_2}{\Vert F_{\Omega}\beta^0\Vert_2}=0.106$. From the model $y=F_{\Omega}\beta^0+\varepsilon$ we know $F_{\Omega}$, the truncated (subsampled) data $y\in\mathbb{C}^{36864}$ and the noise level $\sigma$. The goal is to reconstruct the image $\beta^0$ and to provide confidence intervals for every component of the ground truth. After solving the LASSO, where we chose again $\lambda=25\lambda_0=25\frac{\sigma}{\sqrt{n}}(2+\sqrt{10\log p})$ we derive the desparsified LASSO (displayed in Figure \ref{fig:recov_sparse_angio}) and the confidence regions. Since only the real part is relevant, we obtain intervals, and not regions.

\begin{figure}
\begin{minipage}[b]{1.0\linewidth}
  \centering
  \centerline{\includegraphics[width=8.5cm]{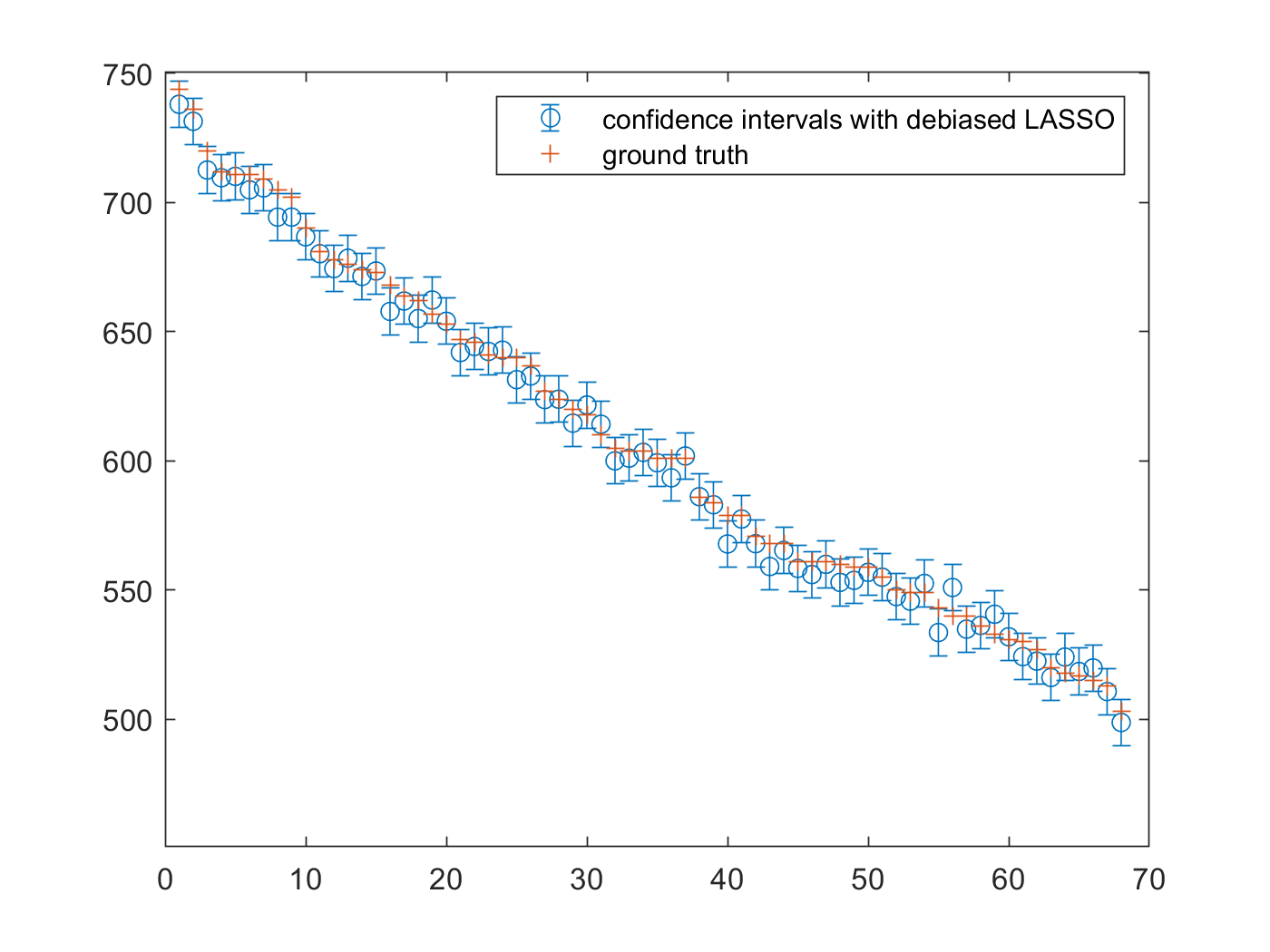}}
\end{minipage}
\caption{Confidence intervals based on the desparsified LASSO for the pixels with largest magnitude sorted in a descending order.}
\label{fig:conf_int_largest_pixels}
\end{figure}

We plot the confidence intervals based on the desparsified LASSO estimator and the ground truth for the 68 pixels with the largest magnitude in Figure \ref{fig:conf_int_largest_pixels}. In almost every case we construct the confidence intervals such, that they contain the ground truth. In order to quantify this behavior we calculate the hitrates for 100 realizations of the subsampled Fourier matrix and the noise. The results are presented in Table \ref{tab:hitrates}. Additionally, we change the threshold [in a.u.] leading to different sparsity levels. Besides the hitrates we provide SSIM coefficients that measure the similarity of the ground truth image $\beta^0$ and the image of the estimator $\hat{\beta}^u$. There is no concrete threshold where the method fails, but the method is rather robust and works better if $s$ is small or if $n$ is large. Even though the (sufficient) condition $n\gtrsim s\log^2 p$ is not fulfilled for any threshold, the practical results show that the UQ procedure still works very well. Note that the hitrates do not depend on the noise level, since the radius of confidence regions in Theorem \ref{thm:confidence circle} scales with the noise level. For example, a threshold of 200 and a noise level of $\sigma=2000$ lead to hitrates of $h_{S_0}=0.932$ and $h=0.951$.

\begin{table}[t]
    \centering
    \begin{tabular}{c|c|c|c|c}
         threshold & $s_0$ & $h_{S_0}$ & $h$ & SSIM \\
         \hline 
         210 & 648 & 0.942 & 0.955 & 0.967 \\
         200 & 1282 & 0.931 & 0.951 & 0.964 \\
         190 & 2789 & 0.901 & 0.941 & 0.954 \\
         180 & 5510 & 0.823 & 0.916 & 0.889
    \end{tabular}
    \caption{Values of $h_{S_0}$, $h$ and SSIM for different sparsity levels $s_0$ and constant $n=0.4 p$. The values are averaged over 100 realizations of $F_{\Omega}$ and $\varepsilon$. The sparsity levels are based on the different image intensity thresholds [a.u.].}
    \label{tab:hitrates}
\end{table}

\section{Conclusion}

In this paper, we derived honest confidence intervals for bounded orthonormal systems. The length of the confidence intervals decreases by the optimal rate $\frac{1}{\sqrt{n}}$. We proved that the amount of data that is sufficient for the uncertainty quantification process is $n \gtrsim \max\{ s\log^2 s\log p, s \log^2 p \}$. For this purpose we desparsified the LASSO estimator paying the price of loosing sparsity.
Our established theory for bounded orthonormal systems includes measurement operators usually employed in MRI, which can be represented by a subsampled Fourier matrix. This highly relevant application motivates us to extend the intervals to two-dimensional confidence regions. In addition we showed that the desparsified LASSO can be extended to the case when the underlying ground truth is not sparse in the canonical basis but sparse with respect to the Haar wavelet basis.

\appendix

\section{Acknowledgment}

We gratefully acknowledge financial support with funds provided by the German Federal Ministry of Education and Research in the grant \ldq SparseMRI3D+: Compressive Sensing und Quantifizierung von Unsicherheiten f\"ur die beschleunigte multiparametrische quantitative Magnetresonanztomografie (FZK 05M20WOA)\rdq.

\bibliography{main}

\section{Appendix}\label{appendix}

In this supplement to the paper, we recall a concentration inequality in Section \ref{subsec:bernstein} and the proof of three theorems in Section \ref{subsec:oracle} that allow us to obtain our main result, namely, Theorem \ref{thm:mainresult}. These are an RIP-based proof of an $\ell_1$ oracle bound, a noise-estimate inequality for BOS, and an RIP-based $\ell_2$ oracle bound. Finally, Section \ref{subsec:proofs_main_body} contains detailed proofs that were skipped in the main body of this paper.

\subsection{Bernstein's Inequality}\label{subsec:bernstein}

Bernstein's inequality for bounded random variables is one of the tools used in our proof.

\begin{theorem}\cite[Theorem 2.8.4.]{Vershynin.2018}\label{bernstein}
Let $Z_1,\hdots,Z_n$ be independent, mean zero random variables, such that $\vert Z_i\vert\leq K$ a.s. for all $i\in[n]$. Then, for every $t\geq 0$, we have
$$\mathbb{P}\left(\left\vert\sum\limits_{i=1}^nZ_i\right\vert\geq t\right)\leq 2\exp\left(-\frac{t^2/2}{\sum\limits_{i=1}^n\mathbb{E}[Z_i^2]+Kt/3}\right).$$
\end{theorem}

\subsection{Oracle Inequalities}\label{subsec:oracle}

We start with an estimate from \cite{FOUCART2023441} of the support size of the LASSO solution $\hat{\beta}$, defined as the minimizer of \eqref{eq:LASSO}. For this purpose we define $\hat{S}=\{i\in[p]\mid \hat{\beta}_i\neq 0\}$. For some constant $\delta\in(0,1)$ and $s_0<p$ we assume that the matrix $\frac{1}{\sqrt{n}}X$ satisfies the RIP of order $t:=36\cdot \frac{1+\delta}{1-\delta}\cdot s_0+1$ and constant $\delta_t<\delta$. It is shown in \cite{FOUCART2023441} that 
\begin{equation}\label{eq:support_LASSO_cardinality}
    \vert\hat{S}\vert\leq 36\cdot \frac{1+\delta}{1-\delta}\cdot s_0:=C_{\delta}\cdot s_0.
\end{equation}
Similar bounds can be found in \cite[Chapter 6]{Buhlmann.2011} or \cite[Chapter 7]{wainwright2019high} when the restricted eigenvalue condition or the compatibility condition is assumed instead of the RIP.

\begin{theorem}\label{thm:LASSOforRIP}
Suppose that $\frac{1}{\sqrt{n}}X$ satisfies the RIP with order $t:=(36\cdot \frac{1+\delta}{1-\delta}+1) s_0$ and constant $\delta_t<\delta<1$. Let further $\max\limits_{1\leq j\leq p}\frac{2}{n}\vert\langle\varepsilon, x_j\rangle\vert\leq\lambda_0$ for some $\lambda_0$. Then, for $\lambda\geq 2\lambda_0$, we have
$$\frac{1}{n}\Vert X(\hat{\beta}-\beta^0)\Vert_2^2+\lambda\Vert\hat{\beta}-\beta^0\Vert_1\leq\frac{4\lambda^2 s_0}{1-\delta_{t}}.$$
In particular, it holds that $\Vert\hat{\beta}-\beta^0\Vert_1\leq\frac{4\lambda s_0}{1-\delta_{t}}$.
\begin{proof}
By \eqref{eq:support_LASSO_cardinality} the difference vector $\hat{\beta}-\beta^0$ is $t=(1+ C_{\delta})s_0$-sparse. The inequality
$$\frac{2}{n}\Vert X(\hat{\beta}-\beta^0)\Vert_2^2+\lambda\Vert\hat{\beta}_{S_0^c}\Vert_1\leq 3\lambda\Vert\hat{\beta}_{S_0}-\beta_{S_0}^0\Vert_1$$
is derived in \cite[Lemma 6.3]{Buhlmann.2011}. Together with the definition of the RIP, we deduce that
\begin{align*}
\frac{2}{n}\Vert X(\hat{\beta}-\beta^0)\Vert_2^2+\lambda\Vert\hat{\beta}-\beta^0\Vert_1&=\frac{2}{n}\Vert X(\hat{\beta}-\beta^0)\Vert_2^2+\lambda\Vert\hat{\beta}_{S_0}-\beta^0_{S_0}\Vert_1+\lambda\Vert\hat{\beta}_{S_0^c}\Vert_1\\
&\leq 4\lambda\Vert\hat{\beta}_{S_0}-\beta^0_{S_0}\Vert_1
\leq 4\lambda\sqrt{s_0}\Vert\hat{\beta}_{S_0}-\beta^0_{S_0}\Vert_2\\
&\leq 4\lambda\sqrt{s_0}\Vert\hat{\beta}-\beta^0\Vert_2
\leq 4\lambda\sqrt{s_0}\frac{1}{\sqrt{1-\delta_{t}}\sqrt{n}}\Vert X(\hat{\beta}-\beta^0)\Vert_2\\
&\leq\frac{1}{n}\Vert X(\hat{\beta}-\beta^0)\Vert_2^2+\frac{4\lambda^2 s_0}{(1-\delta_{t})},
\end{align*}
where we used the inequality $4ab\leq a^2+4b^2$ in the last step. By subtracting $\frac{1}{n}\Vert X(\hat{\beta}-\beta^0)\Vert_2^2$ from both sides, the result follows.
\end{proof}
\end{theorem}

The following lemma, which can be found in \cite[Lemma 6.2]{Buhlmann.2011}, provides an upper bound for the impact of the noise $\varepsilon$ on a deterministic measurement matrix $X$. More precisely, the event $\{\varepsilon\in\mathbb{R}^n:\max\limits_{j\in[p]}\frac{2}{n}\vert\langle\varepsilon,X_j\rangle\vert\leq\lambda_0\}$ holds with high probability. For the sake of completeness, we state and prove a complex version of this result. Please note that a slightly different condition on $\lambda_0$ is obtained as in the real case from \cite[Lemma 6.2]{Buhlmann.2011}.

\begin{lemma}\label{6.2}
Let $\max\limits_{1\leq j\leq p}\hat{\Sigma}_{jj}\leq K$ for some $K>0$. Let further $\varepsilon\sim\mathcal{CN}(0,\sigma^2I_{n\times n})$ be the random noise and $X\in\mathbb{C}^{n\times p}$ a fixed random sampled matrix associated to a BOS. Then for any $t>0$ and for $\lambda_0=\frac{\sigma\sqrt{K}}{\sqrt{n}}(2+\sqrt{t^2+2\log p})$ we have
$$\Pr\left(\max\limits_{j\in[p]}\frac{2}{n}\vert\langle\varepsilon,X_j\rangle\vert\leq\lambda_0\right)\geq 1-2e^{-t^2/4}.$$
In particular, the choice $t^2:=8\log p$ leads to $\lambda_0=\frac{\sigma\sqrt{K}}{\sqrt{n}}(2+\sqrt{10\log p})$ and to the tail bound $\Pr(\max\limits_{j\in[p]}\frac{2}{n}\vert\langle\varepsilon,X_j\rangle\vert\leq\lambda_0)\geq 1-p^{-2}.$
\begin{proof}
We start by noting that $\vert\langle\varepsilon,X_j\rangle\vert=\sqrt{(\Re\langle\varepsilon,X_j\rangle)^2+(\Im\langle\varepsilon,X_j\rangle)^2}=\Vert g_j\Vert_2,$ where $g_j=(\Re\langle \varepsilon,X_j\rangle,\Im\langle\varepsilon,X_j\rangle)^T$ is a two-dimensional normal distributed real random vector. The mean of its components is given by
$$\mathbb{E}[\Re\langle\varepsilon,X_j\rangle]=\langle\mathbb{E}\Re\varepsilon,X_j\rangle=0=\mathbb{E}[\Im\langle\varepsilon,X_j\rangle]$$
and the variance by
\begin{align*}
\var(\Re\langle \varepsilon,X_j\rangle)&
=\var\langle\Re \varepsilon,\Re X_j\rangle+\var\langle\Im \varepsilon,\Im X_j\rangle\\
&=\var\left(\sum\limits_{i=1}^n\Re\varepsilon_i\Re (X_j)_i\right)+\var\left(\sum\limits_{i=1}^n\Im\varepsilon_i\Im (X_j)_i\right)\\
&=\sum\limits_{i=1}^n(\Re X_j)_i^2\var(\Re\varepsilon_i)+\sum\limits_{i=1}^n(\Im X_j)_i^2\var(\Im\varepsilon_i)\\
&=\frac{\sigma^2}{2}\sum\limits_{i=1}^n(\Re X_j)_i^2+(\Im X_j)_i^2=\frac{\sigma^2}{2}\Vert X_j\Vert_2^2.
\end{align*}
The normalized vector with entries $\tilde{g}_j:=\frac{\sqrt{2}}{\sigma\Vert X_j\Vert_2}\cdot g_j$, $j\in[n]$ is standard normal distributed. Since $X$ is a matrix associated to a BOS, we obtain
\begin{equation}\label{eq:bound_column_design}
    \Vert X_j\Vert_2\leq\sqrt{K}\cdot \sqrt{n}
\end{equation}
for every $j\in[p]$. It follows, that
\begin{align*}
\Pr\left(\max\limits_{j\in[p]}\frac{2}{n}\vert\langle\varepsilon,X_j\rangle\vert\geq \lambda_0\right)
&\leq \sum\limits_{j=1}^p\mathbb{P}\left(\frac{2}{n}\vert \langle \varepsilon, X_j\rangle \vert\geq \lambda_0\right)
\leq p\cdot\Pr\left(\frac{2}{n}\Vert g_1\Vert_2\geq\lambda_0\right)\\
&\leq p\cdot\Pr\left(\frac{\sqrt{2}\sigma\sqrt{K}}{\sqrt{n}}\Vert\tilde{g}\Vert_2\geq \lambda_0\right)\\
&\leq p\cdot\Pr\left(\Vert\tilde{g}\Vert_2\geq\sqrt{2}+\sqrt{t^2/2+\log p}\right) \leq pe^{-t^2/4-\log p}=e^{-t^2/4},
\end{align*}
where we applied \cite[Equation 8.89]{Foucart.2013} in the third line.
\end{proof}
\end{lemma}

Now, we are able to state the RIP-based result of the $\ell_2$ oracle inequality.

\begin{lemma}\label{le:l2oracleinequality}
Suppose that the RIP of order $t:=(36\cdot \frac{1+\delta}{1-\delta}+1) s_0$ holds for $\frac{1}{\sqrt{n}}X$ with constant $\delta_t<\delta<1$. Assume further that $\max\limits_{1\leq j\leq p}\frac{2}{n}\vert\langle\varepsilon, x_j\rangle\vert\leq\lambda_0=\frac{\sigma\sqrt{K}}{\sqrt{n}}(2+\sqrt{10\log p})$. Then, for $\lambda\geq 2\lambda_0$ and $p\geq 3$, it holds that
$$\Vert \hat{\beta}-\beta^0\Vert_2\leq C_{\delta}^{\sigma}\sqrt{K}\cdot \frac{\sqrt{s_0\log p}}{\sqrt{n}}$$
with $C_{\delta}^{\sigma}=\frac{\sqrt{640}\sigma}{1-\delta_{(1+C_{\delta})s_0}}$.
\begin{proof}
By \eqref{eq:support_LASSO_cardinality} the difference vector $\hat{\beta}-\beta^0$ is $t=(1+ C_{\delta})s_0$-sparse. From the RIP and Theorem \ref{thm:LASSOforRIP}, it follows that
\begin{align*}
\Vert \hat{\beta}-\beta^0\Vert_2^2\leq \frac{\Vert X(\hat{\beta}-\beta^0)\Vert_2^2}{(1-\delta_t)n}\leq \frac{16\lambda_0^2s_0}{(1-\delta_t)^2}=\frac{16\sigma^2K(2+\sqrt{10\log p})^2s_0}{(1-\delta_t)^2n}.
\end{align*}
For $p\geq 3$, we obtain
$$\frac{16\sigma^2K(2+\sqrt{10\log p})^2s_0}{(1-\delta_t)^2n}\leq\frac{16\sigma^2K(2\cdot\sqrt{10\log p})^2s_0}{(1-\delta_t)^2n}=\frac{640\sigma^2Ks_0\log p}{(1-\delta_t)^2n}.$$
\end{proof}
\end{lemma}

\subsection{Proofs of statements in the main body}\label{subsec:proofs_main_body}
In the following we provide the proofs that are omitted in the main body of this paper. 
We provide proofs of two statements of Section \ref{sec:confidence regions}, i.e. \eqref{eq:confidence interval} and Lemma \ref{C13Re}.

\subsubsection{Proof of Equation \eqref{eq:confidence interval}}\label{subsec:proof_conf_int_valid}
\begin{proof}
We start by estimating the probability that the real part of $\beta_i^0$ belongs to the confidence interval. Recalling that $\delta(\alpha):=\frac{\hat{\sigma}\cdot\hat{\Sigma}_{ii}^{1/2}}{\sqrt{2n}}\Phi^{-1}(1-\alpha/2)$, we obtain
\begin{align*}
&\mathbb{P}\left(\Re(\beta_i^0)\in J_i^{\Re}(\alpha)\right)\\
&=1-\left(\mathbb{P}\left(\Re(\beta_i^0)\leq\Re(\hat{\beta}_i^u)-\delta(\alpha)\right)+\mathbb{P}\left(\Re(\beta_i^0)\geq\Re(\hat{\beta}_i^u)+\delta(\alpha)\right)\right)\\
&=1-\left(1-\mathbb{P}\left(\Re(\hat{\beta}^u_i)-\Re(\beta_i^0)\leq\delta(\alpha)\right)\right)-\mathbb{P}\left(\Re(\hat{\beta}_i^u)-\Re(\beta_i^0)\leq -\delta(\alpha)\right)\\
&=\mathbb{P}\left(\frac{\sqrt{n}(\Re(\hat{\beta}_i^u)-\Re(\beta_i^0))}{\hat{\sigma}/\sqrt{2}\hat{\Sigma}_{ii}^{1/2}}\leq\Phi^{-1}(1-\alpha/2)\right)\\
&-\mathbb{P}\left(\frac{\sqrt{n}(\Re(\hat{\beta}_i^u)-\Re(\beta_i^0))}{\hat{\sigma}/\sqrt{2}\hat{\Sigma}_{ii}^{1/2}}\leq-\Phi^{-1}(1-\alpha/2)\right).
\end{align*}
By taking the limit $n \rightarrow \infty$, Lemma \ref{C13Re} gives
\begin{align*}
\lim\limits_{n\to\infty}\mathbb{P}\left(\Re(\beta_i^0)\in J_i^{\Re}(\alpha)\right)&=\lim\limits_{n\to\infty}\mathbb{P}\left(\frac{\sqrt{n}(\Re(\hat{\beta}_i^u)-\Re(\beta_i^0))}{\hat{\sigma}/\sqrt{2}\hat{\Sigma}_{ii}^{1/2}}\leq\Phi^{-1}(1-\alpha/2)\right)\\
&-\lim\limits_{n\to\infty}\mathbb{P}\left(\frac{\sqrt{n}(\Re(\hat{\beta}_i^u)-\Re(\beta_i^0))}{\hat{\sigma}/\sqrt{2}\hat{\Sigma}_{ii}^{1/2}}\leq-\Phi^{-1}(1-\alpha/2)\right)\\
&=1-(\alpha/2)-(1-(1-\alpha/2))=1-\alpha.
\end{align*}
The argument for the imaginary part follows in an analogous way.
\end{proof}

\subsubsection{Proof of Lemma \ref{C13Re}}\label{subsec:proof_of_noise_consistency}
\begin{proof}
We use the decomposition given by \eqref{eq:key_decomposition} in order to obtain
$$\frac{\sqrt{n}(\hat{\beta}_i^u-\beta_i^0)}{\sigma\hat{\Sigma}_{ii}^{1/2}}=\frac{X_i^*\varepsilon}{\sqrt{n}\sigma\hat{\Sigma}_{ii}^{1/2}}+\frac{R_i}{\sigma\hat{\Sigma}_{ii}^{1/2}}$$
Due to Theorem \ref{thm:mainresult} the random variable 
$$W_i^N:=\frac{X_i^*\varepsilon}{\sqrt{n}\sigma\hat{\Sigma}_{ii}^{1/2}}=\frac{W_i}{\sigma\hat{\Sigma}_{ii}^{1/2}}\sim\mathcal{CN}(0,1).$$
is Gaussian with mean $\mathbb{E}[W_i^N]=\frac{1}{\sigma\hat{\Sigma}_{ii}^{1/2}}\mathbb{E}[W_i]=0$ and variance $\var(W_i^N)=\frac{\var(W_i)}{\sigma^2\hat{\Sigma}_{ii}}=1$. By definition of a complex random vector, $\Re(W_i^N)\sim\mathcal{N}(0,1/2)$ and therefore
$$\Re\left(\sqrt{2}W_i^N\right)\sim\mathcal{N}(0,1).$$
Moreover,
\begin{align*}
\mathbb{P}\left(\frac{\sqrt{n}\Re(\hat{\beta}_i^u-\beta_i^0)}{\hat{\sigma}/\sqrt{2}\hat{\Sigma}_{ii}^{1/2}}\leq x\right)&=\mathbb{P}\left(\frac{\sigma}{\hat{\sigma}}\sqrt{2}\Re\left(W_i^N\right)+\frac{\sqrt{2}\Re(R_i)}{\hat{\sigma}\hat{\Sigma}_{ii}^{1/2}}\leq x\right)\\
&\leq \mathbb{P}\left(\frac{\sigma}{\hat{\sigma}}\sqrt{2}\Re\left(W_i^N\right)+\frac{\sqrt{2}\Re(R_i)}{\hat{\sigma}\hat{\Sigma}_{ii}^{1/2}}\leq x + \epsilon\right) \\
&+ \mathbb{P}\left(\frac{\sqrt{2}\Re(R_i)}{\hat{\sigma}\hat{\Sigma}_{ii}^{1/2}} \leq -\epsilon\right)\\
&\leq\mathbb{P}\left(\frac{\sigma}{\hat{\sigma}}\sqrt{2}\Re\left(W_i^N\right)\leq x+\epsilon\right)+\mathbb{P}\left(\frac{\sqrt{2}\vert R_i\vert}{\hat{\sigma}\hat{\Sigma}_{ii}^{1/2}}\geq \epsilon\right).
\end{align*}

The remainder of the proof follows the same arguments as the real case described in \cite[Lemma 13]{Javanmard.2014}.
\end{proof}



\end{document}